\documentclass{article}
\usepackage{arxiv}
\usepackage[utf8]{inputenc} 
\usepackage[T1]{fontenc}    
\usepackage[table,xcdraw, dvipsnames]{xcolor}
\usepackage[colorlinks=true, linkcolor=blue, urlcolor=blue, citecolor=blue]{hyperref}
\usepackage{url}            
\usepackage{booktabs}       
\usepackage{amsfonts, amsmath, amssymb, bm}       
\usepackage{nicefrac}       
\usepackage{microtype}      
\usepackage{graphicx}
\usepackage{natbib}
\usepackage{doi}
\usepackage{algorithm}
\usepackage{algpseudocode}
\usepackage{tikz}
\usepackage{econometrics}
\usepackage{breqn}
\usepackage{enumerate}
\usepackage{subfig}
\usepackage[autostyle=false, style=english]{csquotes}
\usepackage{booktabs} 
\MakeOuterQuote{"}
\usepackage{amsthm}
\usepackage[toc,page]{appendix}
\usepackage{comment}

\theoremstyle{plain}
\newtheorem{theorem}{Theorem}
\newtheorem{corollary}{Corollary}

\newtheorem{lemma}{Lemma}
\newtheorem{assumption}{Assumption}

\theoremstyle{definition}
\newtheorem{definition}{Definition}
\newtheorem{example}{Example}
\newtheorem{remark}{Remark}

\newcommand{\one}{{\bf 1}}
\newcommand{\matdg}[1]{\left[ #1 \right]_{\mathrm{dg}}}
\newcommand{\fundmat}{\mZ_{\mQ}\,}
\newcommand{\ev}{\mathbb{E}}
\newcommand{\R}{\mathbb{R}}

\title{Multiplicative Path Weights in Networks: From Katz Connectivity to Mean First--Passage Times}

\author{
	\href{https://orcid.org/0009-0007-4945-6120}{\includegraphics[scale=0.06]{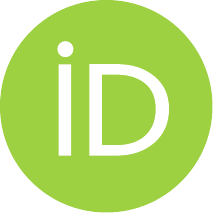}\hspace{1mm}Robbert~van~der~Burg} \\
	Department of Mathematics\\
	Vrije Universiteit Amsterdam\\
	\texttt{\href{mailto:r.vander.burg@vu.nl}{r.vander.burg@vu.nl}} \\
	\And
	\href{https://orcid.org/0000-0002-3389-2311}{\includegraphics[scale=0.06]{orcid.pdf}\hspace{1mm}Bernd~Heidergott} \\
	Department of Operations Analytics\\
	Vrije Universiteit Amsterdam\\
	\texttt{\href{mailto:b.f.heidergott@vu.nl}{b.f.heidergott@vu.nl}}
    \And
    \href{https://orcid.org/0009-0008-1891-2346}{\includegraphics[scale=0.06]{orcid.pdf}\hspace{1mm}Thao~Le}\\
	Department of Economics\\
	Vrije Universiteit Amsterdam\\
	\texttt{\href{mailto:t.p.t.le@vu.nl}{t.p.t.le@vu.nl}} \\
    \And
	\href{https://orcid.org/0000-0001-6585-4785}{\includegraphics[scale=0.06]{orcid.pdf}\hspace{1mm}Alessandro~Zocca} \\
	Department of Mathematics\\
	Vrije Universiteit Amsterdam\\
	\texttt{\href{mailto:a.zocca@vu.nl}{a.zocca@vu.nl}}
}

\renewcommand{\headeright}{ }
\renewcommand{\shorttitle}{Multiplicative Path Weights in Networks}

\hypersetup{
pdftitle={Multiplicative Path Weights in Networks},
pdfauthor={Robbert van der Burg, Bernd Heidergott, Thao Le, Alessandro Zocca}
}

\begin{document}
\maketitle

\begin{abstract}
	Signal dispersion, stress propagation in interbank systems, and transmission probabilities in epidemic and reliability models are all captured by studying paths on directed graphs where the path weight is defined as the product of its edge weights. Connectivity in such networks is naturally quantified by aggregating these multiplicative path weights across nodes. This paper provides a comprehensive treatise on connectivity and ranking measures based on multiplicative path weights. We demonstrate how standard all-path connectivity, such as the standard Katz centrality, can be extended beyond its classical spectral threshold. 
In addition, we derive closed-form expressions for connectivity based on first-hit paths and introduce a new version of Katz centrality.
Furthermore, for random edge weights, we derive resolvent-based expressions for the first and second moments, alongside variance bounds, under a path-indexed coupling framework. 
Finite irreducible Markov chains are a subclass of the models studied, and we show that the scaled matrix of first-hit path probabilities is the probability-generating function for first-passage and return times, and that differentiation yields the classical first-passage times. 
Numerical examples illustrate the application of our results in a variety of settings.
\end{abstract}


\section{Introduction}\label{sec:Intro}
Let $\mathcal G = (\mathcal V, \mathcal E)$ be a finite, directed, strongly connected graph with weights $\mQ(i,j)$ on each edge $(i,j) \in \mathcal E$. The weight of a path is the product of its edge weights and a \emph{first-hit path} from $i$ to $j$ is a path that does not visit $j$ before its final step. We study the \emph{accumulated first-hit weight matrix} $\mM$, whose entries $\mM(i,j) \in [0, +\infty]$ are the sum of the weights of all first-hit paths from $i$ to $j$.%

Our framework can be used in various applications: the expected discounted payoff paid at first entry into a designated regime \citep{hamilton1989new}, the probability that a transmission chain is realized in epidemic and reliability models \citep{wang2003,chakrabarti2008}, and stress propagation in interbank networks, stopped at the first arrival at a monitored creditor. These quantities are not additive, which distinguishes the present framework from \cite{le2026randomwalkstraversalcosts}, in which traversal costs sum along a path. 

When $\mQ$ is row-stochastic, the weight of a path is the probability that a random walker follows it, and $\mM$ becomes the product-accumulated analog of the matrix of mean first-passage times (MFPT). This matrix can be used directly as input to spectral clustering algorithms \citep{fouss2007random,berkhout2019analysis}, whose numerical bottleneck is precisely the cost of inverting a dense deviation matrix (see \cite{HunterMFPT} for an overview).
$\mM$ satisfies the implicit renewal equation
\begin{equation}\label{eq:first_step_decomp}
    \mM(i,j) = \mQ(i,j) + \sum_{k \neq j} \mQ(i,k)\mM(k,j), 
\end{equation}
or, in matrix form,
\begin{equation}\label{eq:basic recursion}
    (\mI-\mQ)\mM = \mQ(\mI - \matdg{\mM}),
\end{equation}
where $\matdg{\cdot}$ denotes the operator that sets all off-diagonal entries to zero. Solving \eqref{eq:basic recursion} for $\mM$ requires inverting $(\mI-\mQ)$, which is not always possible.

Write $\mR(\gamma) := (\mI - \gamma\mQ)^{-1}$ for the $\gamma$-scaled \textit{resolvent matrix}. For $\gamma < 1$ and $\mQ = \mP$ an irreducible, aperiodic Markov chain, $\mR(\gamma)$ equals the Neumann series $\sum_{n=0}^\infty \gamma^n \mP^n$ and is the discounted value function of a Markov reward process,
\begin{equation}\label{eq:MDP}
\sum_{n=0}^\infty \gamma^n  (\mP^n \vr )(i) =
\mathbb{E} \left[ \sum_{n=0}^\infty \gamma^n {\vr}_{X ( n ) } \; \mid \;X_0 = i  \right] ,
\end{equation}
where $\vr$ denotes the reward vector for the Markov chain $ \{ X_n : n \in \mathbb{N} \} $ with transition probability matrix $ \mP$. This probabilistic reading and the convergence of the series itself are lost once $\gamma \geq 1$. For Markov chains, the resolvent has a pole,
\begin{equation}\label{eq:limit1}
\lim_{\gamma \uparrow 1}(1-\gamma) \mR(\gamma) = \lim_{\gamma \uparrow 1}(1-\gamma)(\mI-\gamma\mP)^{-1} = \mPi,
\end{equation}
with $\mPi$ the ergodic projector of $\mP$, see \citep{puterman2014markov}. This is exactly why the classical treatment to study~\eqref{eq:MDP} of the non-discounted case ($\gamma=1$) falls back on the \textit{deviation matrix}, defined implicitly in the right-hand side of
\begin{equation}
(\mI - \mP + \mPi)^{-1} = \mPi + \sum_{n=0}^\infty (\mP^n - \mPi).\nonumber 
\end{equation}

The observation driving this paper is that $\mM(\gamma)$, i.e., the solution of \eqref{eq:basic recursion} with $\mQ$ replaced by $\gamma\mQ$, does \emph{not} inherit this pole; instead, it acts as a removable singularity. In transfer-function language, this is a pole-zero cancellation: the all-walks resolvent diverges exactly at the resonance $\gamma^\ast = 1/\rho(\mQ)$, the boundary treated elsewhere via descriptor systems~\citep{Misra_VanDooren_Syrmos_1995}, average-cost MDPs~\citep{Lamond_Puterman_1989}, and the Laplacian resolvent of network consensus~\citep{Olfati-Saber_Fax_Murray_2007}, yet the first-hit ratio remains finite and analytic straight through it.

The main contributions of this paper are summarized as follows:
\begin{enumerate}
\item We solve equation \eqref{eq:basic recursion} for $\mM$, deriving closed-form solutions in case $(\mI-\mQ)$ is invertible, and on the Perron-critical boundary ($\rho(\mQ)=1$), together with a proof of uniqueness.
\item 
In the Markov case, we identify $\mM(\gamma)$ as the probability-generating function of classical first-passage times and obtain their means by taking its derivative at $\gamma=1$. 

\item We provide simultaneous approximations including the leading error terms to the first-hit matrix, the classical mean first-passage time matrix, the deviation matrix, and the fundamental matrix via a single resolvent evaluation.

\item Centrality ranking based on Katz connectivity~\citep{katz1953new,sharkey2017control}, and more general $\alpha$-connectivity~\citep{Bonacich2001Eigenvector}, ranks a node by the $\alpha$-discounted number of walks that leave it. Typically $\alpha \in (0,1)$, ensuring that paths vanish geometrically fast.
We extend classical Katz connectivity in an interpretable way beyond its walk-sum convergence threshold into $1/\lambda_1<\alpha<1/|\lambda_2|$ using a convergent series of deviations from Perron-asymptotic walk counts together with the Perron term. The resulting scores can be negative and yield different node rankings, while retaining the pole at $\alpha=1/\lambda_1$.

\item 

A well-known feature (see, for instance, \cite{Arrigo_Grindrod_Higham_Noferini_2018}) of classical Katz connectivity is that counting \textit{all} walks includes re-counting self-loops and cycles in the network, inflating the centrality ranking of certain nodes. 
We introduce an alternative version of Katz connectivity (and ranking) based on first-hit paths that is of use in applications that are sensitive to re-counting self-loops.

\item We extend the model to a stochastic setting with random edge weights.
\end{enumerate}

The remainder of the paper is organized as follows. Section~\ref{sec:Model} formally introduces the model. In Sections~\ref{sec:inv} and~\ref{sec:Singular}, solutions to \eqref{eq:basic recursion} are derived in both the invertible and Perron-critical case. Section~\ref{sec:novel} develops a unifying resolvent-based approach. In Section~\ref{sec:Stochastic} we discuss the extension to the stochastic case, followed by Section~\ref{new katz section} which treats Katz connectivity specifically. Section~\ref{sec:Application} contains three applications, followed by a discussion of further research in Section~\ref{sec:Conclusion}.

\section{Graph Model}\label{sec:Model}
We consider a graph ${\cal G} = ({\cal V}, {\cal E})$, with node set ${\cal V}=\{1, \dots, n\}$ and edge set ${\cal E} \subseteq {\cal V}\times {\cal V}$ with $|{\cal E}|=m$. 
Let ${\cal A}$ denote the adjacency matrix of ${\cal G}$, i.e., ${\cal A} (i,j) = 1$ if $(i,j) \in {\cal E}$ and zero otherwise. A calligraphic letter is used to mark its importance and emphasize its relation to the graph ${\cal G}$.
We attach positive weights $\mQ \in \mathbb{R}_{+}^{n\times n}$ to the edges of the graph ${\cal G}$. Throughout the paper, we work with the following setup:

\textit{The graph $\cal G$ is finite, directed and strongly connected, and $\mQ \in \mathbb{R}_+^{n \times n}$ satisfies $\mQ(i,j)>0$ if and only if ${(i,j) \in \cal E}$. Self-loops $(i,i) \in \cal E$ are allowed. Consequently, $\mQ$ is nonnegative and irreducible.}

A path $ \sigma $ on $ {\cal G}$ from $ i$ to $ j $ is a finite sequence 
 ${\sigma = ((i_1,j_1), \dots, (i_l,j_l))}$ of links $ ( i_k , j_k ) \in {\cal E}$, for $ 1 \leq k \leq l $, with $ j_k = i_{k+1}$, for $ 1 \leq k \leq l-1 $, $ i_1 = i$ and $ j_l = j $. 
We denote by $ l (  \sigma )$ the {\em length of path} $ \sigma$, and
\[
 r ( \sigma )= \prod_{k=1}^{l (\sigma) } \mQ ( i_k,j_k ) 
\]
is {\em the weight of a path} $ \sigma$.
A path $ \sigma $ from $ i $ to $ j $ is called a \textit{first-hit path} if $ i_k \not = j$ for $ k=2, \ldots , l$, that is, if $ \sigma $ is a finite path going from $ i $ to $ j $ that does not visit $ j$ before reaching the end node.
We denote by $  {\cal P} ( i , j )  $ the set of all first-hit paths from $ i $ to $ j$, and the accumulated first-hit path weight matrix $ \mM $ is defined by
\begin{equation}\label{eq:accumulated_first_hit_path}
\mM  ( i , j ) := \sum_{ \sigma \in  {\cal P} ( i , j ) }
\prod_{k=1}^{l( \sigma) }  \mQ ( i_k,j_k ) ,
\end{equation}
for $1  \leq i , j \leq n $, with sums interpreted in $[0,+\infty]$. Under this convention, $\mM$ is an extended-real matrix in $\mathbb{R}^{n\times n}$ exactly when all these sums are finite. Decomposition at the first step returns~\eqref{eq:first_step_decomp}, where $0 \cdot(+\infty)=0$. When every entry of $\mM$ is finite, these identities result in~\eqref{eq:basic recursion}.

For two matrices $\mB, \mC$ of the same size, we define the Hadamard product of $\mB$ and $\mC$, $\mB \circ \mC$, as the element-wise product, that is, $(\mB \circ \mC)(i,j) := \mB(i,j) \mC(i,j)$. 
Depending on the choice of $ \mQ $,
our model covers various applications as explained in Table~\ref{tab:placeholder},
where $ \mP $ denotes a Markov transition matrix on $ {\cal G}$ and $ \mW $ an $ n \times n $ real-valued matrix of (possibly stochastic) weights.

\begin{table}[!t]
    \centering
    \caption{Overview of possible interpretations of the weight matrix $\mQ$.}
    \begin{tabular}{c|c|c}
            Model &  $\mQ $ &  $ \mM $  \\ \hline \hline
        scaling of adjacency matrix & $ \alpha {\cal A}$   &  first-hit Katz connectivity   \\
        link weighted graph & $ \mW = \mW \circ {\cal A} $   &   accumulated $ \mW$-weighted first-hit paths   \\
                Markov case &  $ \mP $ & first-hit probabilities\\
        weighted Markov case &  $ \mP \circ \mW $ & multiplicatively weighted first-passage functional      
    \end{tabular}
    \label{tab:placeholder}
\end{table}

\section{The Invertible Case}\label{sec:inv}
In the following, we consider the scaled resolvent $ {( \mI - \gamma \mQ )^{-1}}$ of $ ( \mI - \mQ) $, for $ \gamma  \in \mathbb{R}$.
The following well-known lemma guarantees that there are only finitely many values of $\gamma$ for which $(\mI-\gamma \mQ)$ fails to be invertible. For the proof of the lemma, we refer to \cite[Chapter 7.3]{Meyer2023Matrix}.

\smallskip

\begin{lemma}\label{lem: finite alpha}
    Let $\mQ$ have $ k$ distinct eigenvalues $ \lambda_1, \dots, \lambda_k $, for $k \leq n $, and spectral radius $ \rho ( \mQ )  =  \max_i |\lambda_i|$.
   Then, 
    \begin{enumerate}
        \item $\mI- \gamma \mQ$ is invertible for $\gamma \in \mathbb{R} \setminus \{ 1/ \lambda_1 , \ldots , 1 / \lambda_k \} $,

        \item the Neumann series $ \sum_{n=0}^\infty \gamma^n \mQ^n $ converges for $ {|\gamma| < 1/\rho ( \mQ ) }$.
    \end{enumerate}
\end{lemma}

Lemma~\ref{lem: finite alpha} shows that the inverse of $( \mI - \gamma \mQ )$ exists for almost all values of $ \gamma$, whereas the Neumann series only converges for a strict subset of those values, establishing that invertibility of $ ( \mI - \gamma \mQ ) $ is a strictly weaker condition than the convergence of the Neumann series $ \sum_{n=0}^\infty \gamma^n \mQ^n $. 
For $ \gamma $ in the convergence radius of the Neumann series, the closed form $ ( \mI - \gamma \mQ )^{-1} $ coincides with the discounted walk-count generating function; beyond it, the closed form retains algebraic meaning as the unique analytic continuation of the generating function, although the path-sum interpretation no longer applies.

When $\mQ = \mP$ is a Markov chain, solving \eqref{eq:basic recursion} yields a multiplicative equivalent to the classical \textit{mean first-passage time} matrix, whose $(i,j)$-th entry gives the expected number of steps a random walker takes to reach node $j$ for the first time starting from node $i$.
As the next lemma shows, 
Markov chains $ \mP $ constitute a particular class, as $ ( \mI - \mP ) $ is singular without exception.

\smallskip

\begin{lemma}\label{lem: MC}
If $ \mQ=\mP $ is a Markov chain, then, for $ \epsilon > 0 $ sufficiently small,  $( \mI - \gamma \mP )$ is invertible for $ {\gamma  \in  ( 1 - \epsilon , 1 + \epsilon ) \setminus \{ 1 \} }$, and fails to be invertible at $ \gamma =1 $.
\end{lemma}
\begin{proof}
For $ \gamma =1 $, the null space of $ ( \mI - \mP ) $ contains the vector of length $n $ with all entries equal to one, which shows that $ ( \mI - \mP ) $ is not invertible.
Since there are only finitely many singularities by Lemma~\ref{lem: finite alpha}, including $\gamma=1$, there exists an open punctured neighborhood around $ \gamma =1 $ where $ ( \mI - \gamma \mP ) $ is invertible.
\end{proof}

As Lemma~\ref{lem: finite alpha} shows, $\mI - \mQ $ is typically invertible, with the Markovian case being a prominent exception in Table~\ref{tab:placeholder}.

\subsection{General Solution}
In the invertible case, the solution of \eqref{eq:basic recursion} can be obtained quickly and is presented in the following theorem; its proof can be found in Section~\ref{app:proof_thm_1} of the Appendix.
For ease of notation, we let $\mR$ denote the inverse of $ \mI - \mQ$, i.e.
\begin{equation}\label{eq:resolventgamma}
\mR := (\mI-\mQ)^{-1}.
\end{equation}

\begin{theorem}\label{thm: sol invertible}
    Suppose that $\mM$ is finite and $\mI-\mQ$ is invertible. Then $\mR(j,j)\neq 0$ for all $j \in {\cal V}$, and the solution of \eqref{eq:basic recursion} is given by
    \[
        \mM (i,j) = \frac{\mR(i,j)}{\mR(j,j)},\, i \neq j, \;\;\; \text{and} \;\;\; \mM (i,i) = 1- \frac{1}{\mR(i,i)}
    \]
for all $i \in {\cal V}$, which, in matrix notation, reads as
\[
       {\mM = (\mR - \mI)  ( \matdg{\mR} )^{-1}}.
\]
\end{theorem}

Theorem~\ref{thm: sol invertible} does not require irreducibility of $\mQ$. For ${\mQ=\mP}$, the off-diagonal identity in Theorem~\ref{thm: sol invertible} agrees with the known generating-function identity for the first passage times, see \cite[Section~2.5.2]{aldousfill}.

The following assumption is used from here on, except in Section~\ref{sec:Singular}, where Theorem~\ref{thm:sol_singular} proves finiteness.

\smallskip

\begin{assumption}\label{ass:finite}
The first-hit matrix $\mM$ is finite. When the scaling $\gamma\mQ$ is considered, this is understood to hold for the values of $\gamma$ under consideration.
\end{assumption}

\smallskip

\begin{corollary}
    A sufficient condition for the finiteness of $\mM$ is $\rho(\mQ)<1$, such that the Neumann series $\sum_{n\geq0}\mQ^n$ converges, $\mR$ exists, and every first-hit sum is bounded by the corresponding all-walk sum.
\end{corollary}

Without this assumption, the extended-real identities \eqref{eq:first_step_decomp} still hold, but the matrix algebra is not available, and a real solution of \eqref{eq:basic recursion} need not coincide with the path sums.

\smallskip

\begin{example}
    Take, for instance, the nonnegative, irreducible matrix
    \[
        \mQ = \begin{pmatrix}
            0.6 & 0.6 & 1\\
            0.6& 0.6& 1\\
            1&0&0
        \end{pmatrix}.
    \]
    Direct inversion gives
    \[
        \mR = (\mI-\mQ)^{-1} = \begin{pmatrix}
            -1/3& -1/2 & -5/6\\
            -4/3& 1/2 & -5/6\\
            -1/3& -1/2& 1/6
        \end{pmatrix}.
    \]
    Every cycle avoiding node $3$ has weight below one (the loops at nodes $1$ and $2$ have weight $0.6$, the round trip $1\to2\to1$ has weight $0.36$). Yet closed walks within $\{1,2\}$ may combine these cycles, and the total weight of first-hit paths from $1$ to $3$ of length $r+1$ equals $1.2^{\,r}$. Hence $\mM(i,3)=+\infty$. The formulas of Theorem~\ref{thm: sol invertible}, applied formally, return $\mR(i,3)/\mR(3,3)=-5$ for $i=1,2$ and $1-1/\mR(3,3)=-5$ for $i=3$: a negative value for a sum of positive terms, which nevertheless solves the first-step equations. In contrast, the first column is finite and equals $(4,4,1)^\top$, in agreement with the formulas. The resolvent formulas are therefore only meaningful for path sums when $\mM$ is finite.

\end{example}

We now introduce the resolvent with the intent of extending the results of Theorem~\ref{thm: sol invertible},
\begin{align*}
\mR ( \gamma)  := (\mI- \gamma  \mQ)^{-1},
\end{align*}
where $\mQ$ from \eqref{eq:resolventgamma} is now scaled by a factor $\gamma$. For $\gamma \geq 0$, we write $\mM(\gamma)$, with entries $\mM(i,j;\gamma)$, for the first-hit matrix of $\gamma \mQ$; each entry is a power series in $\gamma$ with nonnegative coefficients.

\smallskip

The pole at $ \gamma = 1 $ is called the {\em resonance threshold} \citep{estrada2010network, van2023graph}; more generally, 
each pole $ \gamma = 1 / \lambda_i $, for $ \lambda_i  $ an eigenvalue of the matrix $\mQ $, is a resonance of the discounted closed-walk count associated with that eigenvalue.
The smallest value for $ \gamma $ at which we encounter a resonance is $ \gamma  = 1 / \rho ( \mQ ) $, with 
$\rho ( \mQ )$ denoting the spectral radius of $ \mQ$, which is equal to one for every stochastic matrix.

\smallskip

\begin{example}
In epidemic models on a graph $\mathcal{G}$ with transmission probability $\gamma \in (0,1)$ per contact, the cumulative expected infections starting from seed $i$ is $\sum_{k=0}^\infty \gamma^k (\mathcal{A}^k \one)_i = (\mR(\gamma)\one)_i$. The product $R_0 := \gamma \rho(\mathcal{A})$ acts as the basic reproduction number \citep{vandendriessche2002}:
\begin{itemize}
    \item $R_0 < 1$: Outbreaks die out, the Neumann series converges and $\mR(i,i;\gamma)$ remains finite.
    \item $R_0 = 1$: The epidemic threshold, $\mR(i,i;\gamma)$ exhibits a pole, reflecting infinite cumulative cases.
    \item $R_0 > 1$: Exponential growth, the series diverges, and $\mR(i,i;\gamma)$ persists only via analytic continuation.
\end{itemize}
The pole at $R_0 = 1$ directly recovers the spectral threshold of linearized SIS dynamics on networks from walk-counting principles \citep{wang2003,chakrabarti2008}.
\end{example}

\section{The Perron-Critical case}\label{sec:Singular}
While the previous results hold in case $\gamma < 1/\rho(\mQ)$ and the $\mI-\gamma\mQ$ is invertible, a critical shift occurs at the boundary. As $\gamma$ approaches 1, the matrix $\mI-\gamma\mQ$ becomes singular, necessitating a distinct analysis in the non-invertible regime. We therefore direct our focus to the \textit{Perron-critical case} where $\rho(\mQ)=1$, in which case $\mI -\mQ$ is singular. By the Perron-Frobenius theorem, eigenvalue one is algebraically simple, with real right and left eigenvectors $\vx, \vy >0$.

\smallskip

\begin{lemma}\label{lem:diagonals non invertible}
    Assume $\rho(\mQ)=1$, and let $\mM$ satisfy Assumption~\ref{ass:finite}. Then, $\matdg{\mM} = \mI$.
\end{lemma}
\begin{proof}
    Continuing from \eqref{eq:basic recursion}, multiplying both sides of the equation on the left by $\vy^\top$ gives \[\vy^\top(\mI - \mQ)\mM = \vy^\top\mQ ( \mI - \matdg{\mM}),\] and combining with $\vy^\top \mQ = \vy^\top$ gives us $
        {\vy^\top (\mI - \matdg{\mM})={\bf 0}}$. 
    Since $\vy>0$, $\mI - \matdg{\mM} = {\bf 0}$, resulting in $\matdg{\mM} = \mI$.
    \end{proof}
In the classical weighted Markovian graph setting, solving $\mM$ requires inverting $\mI-\mP$ through a generalized inverse approach, using, for instance, the fundamental matrix or the group inverse \citep{kemeny1976finite,meyer1975role,hunter1}. This is also the route taken in the additive traversal-cost framework of \cite{le2026randomwalkstraversalcosts}. This same generalized inverse approach can be followed in our multiplicative setting, see Appendix~\ref{sec:gen_inverse_approach_appendix}, but a more direct and elegant approach exists, which we present in the following theorem.

\smallskip

\begin{theorem}\label{thm:sol_singular}
    Assume $\rho(\mQ)=1$, and let $\vx>0$ satisfy $\mQ \vx = \vx$. Then,
    \begin{equation}\label{sol_perron_critical}
        \mM(i,j) = \frac{\vx_i}{\vx_j} \qquad \text{for all } i,j \in {\cal V}.
    \end{equation}
\end{theorem}
\begin{proof}
Define $\widetilde\mP(i,k):=\mQ(i,k)\,\vx_k/\vx_i$. The identity $\mQ\vx=\vx$ shows that every row of $\widetilde\mP$ sums to one, and $\widetilde\mP$ has the support of $\mQ$; it is therefore the transition matrix of a finite irreducible Markov chain $(X_n)$. For a target $j$ let $\tau_{ij}:=\inf\{n\geq1:X_n=j\mid X_0=i\}$. For a finite irreducible chain, $\mathbb P(\tau_{ij}<\infty)=1$ for every $i$, including $i=j$. The events corresponding to first-hit paths from $i$ to $j$ form a disjoint decomposition of $\{\tau_{ij}<\infty\}$, and along each such path, $\sigma$, the telescoping product gives
\[
\prod_{k=1}^{l(\sigma)}\widetilde\mP(i_k,j_k)=\frac{\vx_j}{\vx_i}\prod_{k=1}^{l(\sigma)}\mQ(i_k,j_k).
\]
Summing these nonnegative terms over ${\cal P}(i,j)$ gives
\[1=\mathbb P(\tau_{ij}<\infty)=\frac{\vx_j}{\vx_i}\,\mM(i,j),\]
which proves finiteness and concludes the proof.
\end{proof}

The proof uses neither aperiodicity nor the left Perron vector nor the notion of generalized inverses (which are needed in the additive model). However, irreducibility is essential.

\smallskip

\begin{corollary}\label{cor:uniqueness}
    Under the assumptions of Theorem~\ref{thm:sol_singular}, equation~\eqref{eq:basic recursion} has exactly one real solution, namely the matrix \eqref{sol_perron_critical}.
\end{corollary}
\begin{proof}
    Let $\mM \in \mathbb{R}^{n \times n}$ satisfy \eqref{eq:basic recursion}. Lemma~\ref{lem:diagonals non invertible} gives $\matdg{\mM} = \mI$, so \eqref{eq:basic recursion} reduces to $(\mI-\mQ)\mM=0$. Perron-Frobenius gives $\ker(\mI-\mQ) = {\rm span}\{\vx\}$, which, when combined, proves that every column of $\mM$ is contained in $\ker(\mI-\mQ)$ and the $j$-th column can thus be written as $c_j \vx$ for some $c_j \in \mathbb{R}$. Since the diagonal entries are one, $1 = \mM(j,j) = c_j \vx_j$ and $c_j = 1/\vx_j$. Conversely, the matrix with entries $\vx_i/\vx_j$ has unit diagonal and is fixed column-wise by $\mQ$ and hence satisfies \eqref{eq:basic recursion}. This proves uniqueness.
\end{proof}
For Markov chains specifically, we obtain the following result by noting that for a stochastic matrix $\mP$, the right Perron eigenvector is $\one$.

\smallskip

\begin{corollary}\label{cor:MC_mfpt_product}
    Suppose $\mQ = \mP$ is stochastic, then ${\mM = \one\one^\top}$.
\end{corollary}

The collapse $\mM = \one\one^\top$ for any irreducible and aperiodic Markov chain $\mP$ is not a degeneracy of the formula. The proof of Theorem~\ref{thm:sol_singular} shows that, for a stochastic matrix, $\mM$ returns the total probability of reaching $ j  $ from $i$, which is always $1$ for irreducible chains.

\section{The Resolvent as a Unifying Analytical and Computational Tool}\label{sec:novel}
We introduce the following assumption and its immediate corollary, which we justify using a power series argument.

\smallskip

\begin{assumption}\label{ass:A1}
$\mQ$ is primitive with eigenvalues ${ |\lambda_1| >  |\lambda_2| \geq \ldots \geq  |\lambda_k|}$, spectral radius $ \lambda := \lambda_1 =\rho(\mQ)$, right and left Perron eigenvectors $\vx > 0$ and $\vy > 0$ normalized so $\vy^\top \vx = 1$, and rank-one Perron projector $\mathbf{H} := \vx \vy^\top$.
\end{assumption}
\smallskip
\begin{corollary}\label{ass:A2}
Let $\mQ$ satisfy Assumption~\ref{ass:A1}. Then there exist constants $c>0$ and $r \in (|\lambda_2|, \lambda_1)$ such that
${\|{\mQ}^n - \lambda_1^{\,n}\mathbf{H}\|} \leq c r^n$ for all $n \geq 0$.%
\end{corollary}
\begin{proof}
The matrix $A={\mQ}-\lambda_1 \mathbf{H}$ has spectral radius $|\lambda_2|$. For every $r\in(|\lambda_2|,\lambda_1)$, Gelfand's formula gives $\|A^n\|\leq cr^n$ for a sufficiently large constant $c$. Since ${\mQ}^n-\lambda_1^n \mathbf{H}=A^n$ for $n\geq1$, enlarging $c$ once more covers $n=0$, where the term is $I-\mathbf{H}$.
\end{proof}

If the second eigenvalue is $0$, we set $1/\lambda_2 = \infty$.
We denote by $ {\cal D} (r ) $ the disc 
$ \{ z \in \mathbb{C} : |z| < r \}$ in the complex plane.
Consider the following three objects:
\begin{itemize}
\item
The matrix {\em direct power series}
\[
 \mK (x)  = \sum_{n=0}^\infty \mQ^n x^n 
\]
with radius of convergence $1/\lambda_1$, so that $ \mK (x)$ is holomorphic on the open disc $ {\cal D} ( 1/\lambda_1 ) $.

\item
The {\em resolvent of $ \mQ $}, given by
\[
\mR ( x ) = ( \mI - x \mQ)^{-1} ,
\]
which is holomorphic on $ \mathbb{C} $ except at the reciprocal eigenvalues $ \{ 1/\lambda_i : 1 \leq i \leq k \} $.

\item 
The {\em deviation  power series}, given by
\[
 \mJ ( x )  = \sum_{n=0}^\infty ( \mQ^n - \lambda_1^{\,n}\mH ) x^n + \frac{1}{1-\lambda_1 x}  \mH  ,
 \] 
which is holomorphic on $ {\cal D} ( 1/ |\lambda_2| ) \setminus \{ 1/\lambda_1 \} $ provided that Assumption~\ref{ass:A1} holds, as Corollary~\ref{ass:A2} shows.
\end{itemize}

By simple algebra it follows on $ {\cal D}( 1/\lambda_1 ) $:
\begin{align}\label{eq:transform}
&( \mI - x \mQ)^{-1}   = \sum_{n=0}^\infty \mQ^n x^n   \nonumber\\
&\qquad = \sum_{n=0}^\infty ( \mQ^n - \lambda_1^{\,n}\mH ) x^n 
+ \sum_{n=0}^\infty \lambda_1^{\,n}\mH  x^n 
\nonumber \\
&\qquad = \sum_{n=0}^\infty ( \mQ^n - \lambda_1^{\,n}\mH ) x^n 
+ \frac{1}{ 1 - \lambda_1 x }  \mH  ,
\end{align}
for $ x \in {\cal D}( 1/\lambda_1 ) $.

The resolvent and deviation power series are holomorphic in $ {\cal D} ( 1 / |\lambda_2| )  \setminus \{ 1 /\lambda_1 \} $ and identical on $ {\cal D} ( 1 / \lambda_1)$, so by the identity theorem for holomorphic functions, we have 
\[
\mR ( x )  = \mJ ( x ),\]
for $x \in  {\cal D} ( 1/|\lambda_2| ) \setminus \{  1 / \lambda_1 \}.$
From (\ref{eq:transform}) it now follows that 
\[
\lim_{ x \rightarrow 1/ \lambda_1 } ( 1 - \lambda_1 x ) \mR ( x )  = \mH ,
\]
so that $ \mR $ has a simple pole at $ x = 1/\lambda_1 $ with residue $ -\mH/\lambda_1 $. The singularity at $ x=1/\lambda_1 $ becomes removable in the first-hit ratio $ \mM(i,j) = \mR(i,j)/\mR(j,j) $ for $i\neq j$, where the polar parts cancel. We summarize our analysis in Table~\ref{tab:fgh}.

\begin{table}[!t]
    \centering
    \caption{Relations between the resolvent, Neumann series, and deviation series}
    \begin{tabular}{c|c}
    relation & domain \\
    \hline 
   $ \mR ( x )  =  \mK (x)  = \mJ ( x ) $      &  $ x \in {\cal D}( 1 / \lambda_1)  $ \\
   $ \mR ( x )  = \mJ ( x ) $      &  $ x \in {\cal D}( 1/ |\lambda_2| ) \setminus \{1 / \lambda_1\}  $ \\
    $ \mR ( x )  $   &
   $ | x | \geq 1 / |\lambda_2|$ and $ x \not \in \{ 1/\lambda_i \}$
    \end{tabular}
    \label{tab:fgh}
\end{table}

\subsection{Mean First-Passage Times: The Resolvent Approach}\label{sec:res}
We have seen that in the Markov case, $\mM= \one\one^\top$ carries no useful information. However, the classical first-passage time information lies in the behavior of $\mM(\gamma)$ near its singularity at $\gamma=1$, which this subsection develops.

Throughout this subsection $\mP=\mQ$ is the transition matrix of a finite, irreducible Markov chain $(X_n)$ on ${\cal V}$ that, for ease of exposition, is assumed to be aperiodic. Let $\vpi$ be its unique stationary distribution, i.e., $\vpi \mP = \vpi$, let $\mPi = \one \vpi^\top$ denote the \textit{ergodic projector}, and let $\mZ = (\mI-\mP+\mPi)^{-1}$ be the canonical fundamental matrix, which exists for every irreducible chain \citep{kemeny1976finite}. The \textit{deviation matrix} $\mD := \mZ- \mPi$ is the group inverse of $\mI-\mP$ \citep{meyer1975role}; it satisfies
\[(\mI-\mP)\mD=\mD(\mI-\mP)=\mI-\mPi
\]and
\[\mPi\mD=\mD\mPi={\bf0}.
\]
If $\mP$ is aperiodic, $\mD=\sum_{n\geq0}(\mP^n-\mPi)$. For a target $j$ and every initial state $i$ let
\[\tau_{ij}:=\inf\{n\geq1:X_n=j \mid X_0 = i\},\] and define
\[
\mL(i,j):=\mathbb E[\tau_{ij}] ,
\]
so that $\mL(j,j)=1/\vpi(j)$ is the mean return time.

Since $\mP^n-\mPi$ is bounded in $n$, the deviation series
\begin{equation}\label{power series D gamma}
\mD(\gamma):=\sum_{n\geq0}\gamma^n\big(\mP^n-\mPi\big)
\end{equation}
converges for $|\gamma|<1$. Using $\mP^n\mPi=\mPi$ for all $n\geq0$, the resolvent $\mR(\gamma)=(\mI-\gamma\mP)^{-1}=\sum_{n\geq0}\gamma^n\mP^n$ splits as
\begin{equation}\label{eq:laurent}
\mR(\gamma)=\mD(\gamma)+\frac{1}{1-\gamma}\,\mPi,\qquad |\gamma|<1 .
\end{equation}
The deviation series has a closed form that extends it to a neighborhood of $\gamma=1$. Since $\mD(\gamma)$ is a power series in $\mP$, it commutes with $\mP$ and $\mD$, and $\mD(\gamma)\mPi=\mPi\mD(\gamma)={\bf0}$. Premultiplying \eqref{eq:laurent} by $\mI-\gamma\mP$ and using $(\mI-\gamma\mP)\mPi=(1-\gamma)\mPi$ gives \[(\mI-\gamma\mP)\mD(\gamma)=\mI-\mPi.\] Multiplying this on the right by $\mD$, with the expansion $\mI-\gamma\mP=(1-\gamma)\mI+\gamma(\mI-\mP)$, gives
\[
\mD=\mD(\gamma)\big[(1-\gamma)\mD+\gamma(\mI-\mPi)\big]=\mD(\gamma)\big(\gamma\mI+(1-\gamma)\mD\big).
\]
Hence, for $\gamma \uparrow 1$,
\begin{align}\label{eq:dev-closed}
\mD(\gamma)&=\mD\big(\gamma\mI+(1-\gamma)\mD\big)^{-1}\nonumber\\
&=\mD+(1-\gamma)\big(\mD-\mD^2\big)+{\cal O}\big((1-\gamma)^2\big),
\end{align}
so $\mD(\gamma)$ is analytic at $\gamma=1$ with $\mD(1)=\mD$. Note that, crucially, $\gamma$ should approach one from below to ensure \eqref{power series D gamma} converges. 

From \eqref{eq:laurent} and \eqref{eq:dev-closed},
\begin{align}\label{eq:overall}
\lim_{\gamma\uparrow1}(1-\gamma)\mR(\gamma)=\lim_{\gamma\uparrow1}\Big((1-\gamma)\mD(\gamma)+\mPi\Big)&=\mPi,\\
\qquad
\lim_{\gamma\uparrow1}\Big(\mR(\gamma)-\frac{\mPi}{1-\gamma}\Big)&=\mD .
\end{align}
A first observation is that, although $(1-\gamma)\mR(\gamma)$ admits a continuous continuation at $\gamma=1$, the resolvent itself does not exist there, so Theorem~\ref{thm: sol invertible} cannot be applied at $\gamma=1$. The first-hit ratios nevertheless have a limit. For $0\leq\gamma<1$ and all $i,j\in\mathcal V$,
\begin{align}\label{eq:sm}
\lim_{\gamma\uparrow1}\mM(i,j;\gamma)&=\lim_{\gamma\uparrow1}\frac{(1-\gamma)\mR(i,j;\gamma)}{(1-\gamma)\mR(j,j;\gamma)}\nonumber\\
&=\frac{\mPi(i,j)}{\mPi(j,j)}=1\quad(i\neq j),\nonumber\\
\lim_{\gamma\uparrow1}\mM(j,j;\gamma)&=\lim_{\gamma\uparrow1}\Big(1-\frac{1}{\mR(j,j;\gamma)}\Big)=1,
\end{align}
in agreement with Corollary~\ref{cor:MC_mfpt_product}: at $\gamma=1$, $\mM=\one\one^\top$ records only that every state is hit with probability one.

The classical mean first-passage time matrix $ \mL $, 
where 
$ 
\mL (i,j) = (\mZ(j,j) - \mZ(i,j))/\vpi(j)
$ 
 with $i\neq j$, can be obtained as a derivative of $  \mM ( \gamma )$ at $ \gamma =1 $.
We provide two proofs for the theorem, which can be found in Section~\ref{app: proof theorem 3} of the Appendix. The first proof is direct and algebraic, but fails to provide insight. 
The other is more elaborate and is based on probabilistic reasoning that connects the statement of the theorem to the probability-generating function of the first-passage times.

\smallskip

\begin{theorem}\label{thm:sdq}
    Let $\mP$ be the transition matrix of a finite irreducible Markov chain. For all $i,j \in {\cal V}$, 
\begin{align}\label{eq:mfpt}
\left.\frac{d}{d\gamma }\mM(i,j;\gamma )\right|_{\gamma=1^-}&=\mL(i,j)\nonumber\\
&=
\begin{cases}
\dfrac{\mZ(j,j)-\mZ(i,j)}{\vpi(j)}, & i\neq j,\\[8pt]
\dfrac{1}{\vpi(j)}, & i=j.
\end{cases}
\end{align}
\end{theorem}

\begin{remark}
Since $\mM=\one\one^\top$ for every finite irreducible chain (by Corollary~\ref{cor:MC_mfpt_product}), the doubly weighted form $\sum_{i,j}\vpi(i)\mM(i,j)\vpi(j)=1$ carries no information about $\mP$. The information sits in the $\gamma$-derivative at $\gamma=1$: by Theorem~\ref{thm:sdq},
\[
\sum_j\vpi(j)\,\left.\frac{d}{d\gamma }\mM(i,j;\gamma )\right|_{\gamma=1^-}=\sum_j\vpi(j)\,\mL(i,j) = K
\]
where $K$ is the Kemeny constant, which does not depend on $i$ \citep{kemeny1976finite}.
\end{remark}
Theorem~\ref{thm:sdq} provides a route to computing the classical mean first-passage time $\mL$ without computing the fundamental matrix $\mZ$ or solving any system involving the singular matrix $\mI - \mP$. One simply evaluates $\mM (\gamma ) $ at $\gamma = 1 - \epsilon$, for $\epsilon>0$ and small, requiring only a sparse inversion of $\mI - \gamma \mQ$, and the following approximation can be obtained:
\begin{equation}\label{eq:fd-expansion}
\frac{1-\mM(i,j;1-\epsilon)}{\epsilon}=\mathbb E\tau_{ij}-\frac{\epsilon}{2}\,\mathbb E\big[\tau_{ij}(\tau_{ij}-1)\big]+{\cal O}(\epsilon^2),
\end{equation}
as $\epsilon \downarrow 0$, such that,
\[
    \hat \mL(i,j ;\epsilon) := \frac{1-\mM(i,j;1-\epsilon)}{\epsilon} \rightarrow \mL(i,j).
\]
Note in \eqref{eq:fd-expansion} the appearance of the second factorial moment of the first-hit times. Similarly, \eqref{eq:laurent} yields the following asymptotic identities for the deviation and fundamental matrices.

\begin{lemma}\label{prop: D and Z num}
Let $\mP$ be the transition matrix of a finite irreducible Markov chain, and define $\hat\mD(\epsilon):=\mR(1-\epsilon)-\mPi/\epsilon$ and $\hat\mZ(\epsilon):=\hat\mD(\epsilon)+\mPi$. As $\epsilon\downarrow0$,
\[
\hat\mD(\epsilon)=\mD+\epsilon(\mD-\mD^2)+{\cal O}(\epsilon^2)
\]
and
\[
\hat\mZ(\epsilon)=\mZ+\epsilon(\mD-\mD^2)+{\cal O}(\epsilon^2).
\]
\end{lemma}
\begin{proof}
Set $\gamma=1-\epsilon$ in \eqref{eq:laurent} and use $\mZ=\mD+\mPi$.
\end{proof}

A single inversion of $\mI - (1-\epsilon) \mP$ simultaneously yields an approximation for $\mM$ via Theorem~\ref{thm: sol invertible}, an approximation for $\mL$ via Theorem~\ref{thm:sdq} and approximations for both $\mD$ and $\mZ$ via Lemma~\ref{prop: D and Z num}. 
Note that Lemma~\ref{prop: D and Z num} relies on computing the Perron projector $\mPi$. Since, in our setup, chains are assumed to be irreducible, $\vpi$ can be solved as a linear system via the \textit{Grassmann-Taksar-Heyman algorithm} \citep{Grassmann_Taksar_Heyman_1985}, with complexity ${\cal O}(n^3)$, and $\mPi$ can simply be constructed as the matrix with $\vpi$ on each of its rows. Alternatively, one may consider more efficient iterative methods such as \textit{power iteration}, which reduces the cost to ${\cal O}(k \cdot m)$, where $k$ is the iteration count. However, it should be noted that the convergence of iterative methods could be slow if the spectral gap is small.

\section{Extension to Stochastic Weights}\label{sec:Stochastic}
We treated the weight matrix $\mQ$ as a fixed deterministic object. However, in many applications of practical interest, the edge weights themselves can be random variables: link decays in communication networks fluctuate and per-period discount factors in regime-switching models are drawn from distributions. This section extends the framework to accommodate such randomness and shows that the resolvent approach developed in Section~\ref{sec:novel} extends cleanly to this setting via a \textit{moment-stacking} property that has no analog in \cite{le2026randomwalkstraversalcosts}.

We denote the matrix of stochastic edge weights by ${\cal W} = ({\cal W}(i,j))_{(i,j) \in {\cal E}}$ and define the matrix $\mW(i,j) := \ev[{\cal W}(i,j)]$ through the expected value of the stochastic entries. To handle random weights, we augment the graph model of Section~\ref{sec:Model} with the following assumptions.
\smallskip
\begin{assumption}\label{ass:mean and sec moment}
For every $(i,j)\in\mathcal E$, the weight ${\cal W}(i,j)$ is almost surely positive and has finite mean $\mW(i,j)=\ev[{\cal W}(i,j)]$ and finite second moment $\mW^{(2)}(i,j):=\ev[({\cal W}(i,j))^2]$; ${\cal W}(i,j)=0$ for $(i,j)\notin\mathcal E$. Statements involving the $r$-th moment for $r\geq3$ additionally assume $\ev[({\cal W}(i,j))^r]<\infty$.
\end{assumption}

Assumption~\ref{ass:mean and sec moment} specifies only the edgewise marginals. The random first-hit weight is not determined until the joint law of the draws is fixed, and different couplings lead to different variances. We use the following \emph{path-indexed} coupling.

\smallskip

\begin{assumption}\label{ass:edge_share}
For every $i,j\in\mathcal V$ and every path $\sigma\in{\cal P}(i,j)$ of length $l$, let
$
Y_\sigma:=\prod_{k=1}^{l}{\cal W}(i_k, j_k)
$
denote the stochastic weight of $\sigma$. The realizations $\{{\cal W}(i_1, j_1), \dots, {\cal W}(i_l, j_l)\}$ are independent for each fixed path.
\end{assumption}

Note that Assumption~\ref{ass:edge_share} allows dependence between path weights to arise through shared realizations of stochastic edge weights, whenever common edges in distinct paths are traversed in the same time step. Hence, paths that share edge traversals may be dependent, while independently realized path weights are included as a special case.

Assumptions~\ref{ass:mean and sec moment} and~\ref{ass:edge_share} allow the expectation of the $r$-th power of a single path's weight to factorize as a product of $r$-th moment weights along that path,
\begin{equation}\label{eq:expectation_of_product}
    \ev\left[\left(\prod_k {\cal W}(i_k, j_k)\right)^r\,\right] = \prod_k \ev[  ( {\cal W}(i_k, j_k) ) ^r],
\end{equation}
where $r \geq 1$.

Assumptions~\ref{ass:mean and sec moment} and \ref{ass:edge_share} together generalize the setting of Section~\ref{sec:Model} to one in which 
$\mQ = \mW$ is the matrix of \textit{mean} weights, while higher moments of $\cal W$ enter the analysis for the variance and beyond. Let $\mQ^{(r)} := \mW^{(r)}$ be the matrix formed by taking the $r$-th moment of each weight entry element-wise,
with $\mW = \mW^{(1)}$. Note that when $\mW$ is deterministic, $\mW^{(r)}(i,j) = (\mW(i,j))^r$ is simply the $r$-th (element-wise) power of the weight. 

Define the random total first-hit weight,
\begin{equation}\label{eq:random_hit_time}
    \mT(i,j) = \sum_{\sigma \in {\cal P}(i,j)} Y_\sigma,
\end{equation}
where $Y_\sigma$ is as in Assumption~\ref{ass:edge_share}, the (random) weight realized along a path $\sigma$. Since all terms are nonnegative, expectation and summation may be interchanged and it is now easily seen that
\[
    \mM(i,j) := \ev[\mT(i,j)] = \ev\left[ \sum_{\sigma \in {\cal P}(i,j)} Y_\sigma\right] = \sum_{\sigma \in {\cal P}(i,j)} \ev[ Y_\sigma],
\]
which, combined with \eqref{eq:expectation_of_product} and setting $\mQ = \mW$, recovers \eqref{eq:accumulated_first_hit_path}.

By squaring~\eqref{eq:random_hit_time} and taking expectations, we obtain an expression for the second moment,
\begin{equation}\label{eq:second_moment_hit_times}
    \mM^{(2)}(i,j) := \ev[( \mT(i,j))^2] = \! \! \!\sum_{\sigma \in {\cal P}(i,j)} \! \! \ev[Y_\sigma^2] + \! \sum_{\sigma \neq \sigma'} \ev[Y_\sigma Y_{\sigma '}].
\end{equation}
Note that this expression contains a second moment in itself; namely, the path's own second moment $\ev[Y_\sigma^2]$. We define $\mN^{(r)}$ as the sum of $r$-th moments of path-weights, that is,
\[
    \mN^{(r)}(i,j) := \sum_{\sigma \in {\cal P}(i,j)} \ev[Y_\sigma^r] =\sum_{\sigma \in {\cal P}(i,j)} \prod_{k=1}^{l(\sigma)} \mQ^{(r)}(i_k, j_k).
\]
Each of $\mM=\mN^{(1)}$, $\mN^{(2)}$ and the matrix $\mS$ below is the first-hit matrix of  $\mQ$, $\mQ^{(2)}$ and $\widetilde\mQ$, respectively. The first-step decomposition formula and Theorems~\ref{thm: sol invertible} and~\ref{thm:sol_singular} therefore apply to each of them separately, and in particular the first-step decomposition gives, for finite $\mN^{(r)}$,
\begin{equation}\label{eq:higher order recursion}
    (\mI- \mQ^{(r)})\mN^{(r)} = \mQ^{(r)} (\mI-\matdg{\mN^{(r)}}).
\end{equation}
We call this the \emph{moment-stacking property}. Theorem~\ref{thm:sol_singular} applies to $\mQ^{(r)}$ only if $\rho(\mQ^{(r)})=1$. Moreover, finiteness of $\mN^{(r)}$ is assumed to hold.

We now turn to the second term, where we first show below that Assumption~\ref{ass:edge_share} implies that the path weights are non-negatively correlated. For the proof of the lemma, we refer to Appendix~\ref{app:pos_cor_lem}.

\smallskip

\begin{lemma}\label{lem:pos_corr}
    Let $\sigma, \sigma' \in {\cal P}(i,j)$ be two distinct paths. Under Assumption~\ref{ass:edge_share}, $\ev[Y_\sigma Y_{\sigma'}] \geq \ev[Y_\sigma] \ev[Y_{\sigma'}]$ with equality when weight realizations are independent.
\end{lemma}

Note that we can write
\begin{align*}
    \sum_{\sigma \neq \sigma '} \ev[Y_\sigma]\ev[Y_{\sigma '}]
    &= \left(\sum_{\sigma \in {\cal P}(i,j)} \ev[Y_\sigma]\right)^2 - \sum_{\sigma \in {\cal P}(i,j)} \ev^2[Y_\sigma] \nonumber\\
    &= ( \mM(i,j))^2 - \sum_{\sigma \in {\cal P}(i,j)} \ev^2[Y_\sigma],\nonumber
\end{align*}
and observe that
\[
    \mS(i,j) := \sum_{\sigma \in {\cal P}(i,j)} \ev^2 [Y_\sigma] = \sum_\sigma \prod_k ( \mQ(i_k, j_k) )^2,
\]
which has exactly the same form as~\eqref{eq:accumulated_first_hit_path}, but with $\mQ$ replaced by $\widetilde \mQ := \mQ \circ \mQ$.
We write $\mR^{(2)} := (\mI- \mQ^{(2)})^{-1}$ and $\mR_{\widetilde \mQ} := (\mI - \widetilde \mQ)^{-1}$ whenever these inverses exist.

Therefore, by Lemma~\ref{lem:pos_corr}, \eqref{eq:second_moment_hit_times} rewrites as
\[
    \mM^{(2)}(i,j) \geq \mN^{(2)}(i,j) + (\mM(i,j))^2 - \mS(i,j),
\]
for all $i,j \in {\cal V}$.

Using the solutions for $\mN$ and $\mS$, we get the following result on the variance of the first-hit weights, which we define as
\[
\mV(i,j):=\var[\mT(i,j)].
\]

\smallskip

\begin{theorem}\label{thm:variance}
Let Assumptions~\ref{ass:mean and sec moment} and~\ref{ass:edge_share} hold and suppose that $\mM$ and $\mN^{(2)}$ are finite. Then,
\[
\mV(i,j) \geq\mN^{(2)}(i,j)-\mS(i,j),
\]
where the matrix $\mN^{(r)}$ is given by
\[
\mN^{(r)} = \left((\mI - \mQ^{(r)})^{-1} -\mI\right) \left(\matdg{(\mI - \mQ^{(r)})^{-1}}\right)^{-1}
\]
for $r=1,2$, and by the same formulas with $\widetilde\mQ$ for $\mS$. If, in addition, $\mI-\mQ$, $\mI-\mQ^{(2)}$ and $\mI-\widetilde\mQ$ are invertible, then
\begin{align*}
\mV(i,j)&\geq\frac{\mR^{(2)}(i,j)}{\mR^{(2)}(j,j)}-\frac{\mR_{\widetilde\mQ}(i,j)}{\mR_{\widetilde\mQ}(j,j)}\ (i\neq j),\\
\mV(j,j)&\geq\frac{1}{\mR_{\widetilde\mQ}(j,j)}-\frac{1}{\mR^{(2)}(j,j)}.%
\end{align*}%
\end{theorem}%
\begin{proof}
Since $\widetilde\mQ=\mQ\circ\mQ\leq\mQ^{(2)}$ entrywise by Jensen's inequality and first-hit sums are nondecreasing in the edge weights, $\mS\leq\mN^{(2)}<\infty$. Since $\mM$, $\mN^{(2)}$ and $\mS$ are finite, we can apply the results of Theorem~\ref{thm: sol invertible}. By Tonelli's theorem and Lemma~\ref{lem:pos_corr},
\begin{align*}
\ev[\mT(i,j)^2]&\geq\sum_{\sigma}\ev[Y_\sigma^2]+\sum_{\sigma\neq\sigma'}\ev[Y_\sigma]\,\ev[Y_{\sigma'}]\\
&=\mN^{(2)}(i,j)+\mM(i,j)^2-\mS(i,j),
\end{align*}
where all terms are finite, so $\var[\mT(i,j)]\geq\mN^{(2)}(i,j)-\mS(i,j)$. The resolvent expressions follow from Theorem~\ref{thm: sol invertible} applied to the finite columns of $\mN^{(2)}$ and $\mS$.
\end{proof}

\begin{remark}
    The dependence structure in Assumption~\ref{ass:edge_share} admits two natural interpretations. If the quantity propagated through the network is subject to a common stochastic environment, then paths sharing an edge traversal may also share the corresponding realization of its stochastic weight, resulting in positive dependence between their path weights. This is natural, for example, when a signal or flow branches after having experienced the same stochastic attenuation.

    Alternatively, if paths represent independently transmitted packets or repeated realizations of the network process, each path may experience an independent draw of the stochastic weights, even when the underlying paths overlap. In this case the path weights are independent, and  Lemma~\ref{lem:pos_corr} and Theorem~\ref{thm:variance} both hold with equality.
\end{remark}

\smallskip

\begin{remark}
    Lemma~\ref{lem:pos_corr} bounds the cross moments from below. A matching upper bound holds
without any assumption on the dependence between paths. For all
$\sigma,\sigma'\in {\cal P}(i,j)$, the Cauchy--Schwarz inequality gives $\ev[Y_\sigma Y_{\sigma'}]\;\le\;\|Y_\sigma\|_2\,\|Y_{\sigma'}\|_2$, where
\[  \|Y_\sigma\|_2=\prod_{k=1}^{l}\sqrt{\mQ^{(2)}(i_k,i_{k+1})}.
\]
Write $\hat\mQ:=\sqrt{\mQ^{(2)}}$ for the entrywise square root, so that
$\sum_{\sigma\in P(i,j)}\|Y_\sigma\|_2=\mM_{\hat\mQ}(i,j)$. Summing over all pairs of paths yields
$\ev\big[\mT(i,j)^2\big]\le\mM_{\hat\mQ}(i,j)^2$, which gives the upper bound
\[
  \var\big(\mT(i,j)\big)\;\le\;\mM_{\hat\mQ}(i,j)^2-\mM(i,j)^2.
\]
\end{remark}

\section{First-Hit and Extended Katz Connectivity}\label{new katz section}
The classical Katz connectivity of a network is expressed
through the sum over all discounted walks
\[\sum_{k=1}^{\infty} (\mathcal A^k)(i,j)\alpha^k,\]
and the
outgoing Katz centrality is defined as
\[\vc_{\rm Katz}(i;\alpha) = \sum_{j\in\mathcal V}\sum_{k=1}^{\infty} (\mathcal A^k)(i,j)\alpha^k,\]
where $i \in {\cal V}$ and $0\leq\alpha<1/\rho(\mathcal A)$. Here $\alpha$ is the attenuation factor, replacing $\gamma$. As $\alpha\downarrow0$, $\vc_{\rm Katz}(\alpha)\to 0$. As $\alpha\uparrow1/\lambda_1$, where $\lambda_1=\rho(\mathcal A)$, $(1-\alpha\lambda_1)\vc_{\rm Katz}(\alpha)\to\vx(\vy^\top\one)$, proportional to the right Perron eigenvector. Katz connectivity counts all walks, including first-hit paths from $i$ to $j$ concatenated with subsequent closed walks from $j$ to itself. Counting these revisits can overstate influence when the intended quantity is first-arrival exposure. Recently, \cite{Arrigo_Bertaccini_Filippo_2025} introduced ${\mathcal{F}}$-avoiding first-passage walks, which reach a target node for the first time while avoiding a prescribed set of nodes $\mathcal{F}$. In contrast, we study accumulated weights of paths, based on a resolvent representation.

\subsection{First-Hit Katz connectivity}
We now modify the Katz connectivity.

\smallskip

\begin{definition}
    Let $\mQ=\alpha{\cal A}$ with $\alpha\geq0$.
Then 
\[
\mM ( \alpha ; i , j ) = \sum_{ \sigma \in {\cal P}( i , j )} \alpha^{ l(\sigma ) }
\]
is the accumulated first-hit path weight from $i$ to $j$.
The {\em first-hit Katz connectivity} of $ {\cal A}$ is given by $ \mM ( \alpha ) $, and the {\em first-hit Katz centrality} of node $ i$ is given by
\begin{equation}\label{eq:FHKatz}
\vc_{\rm FHKatz} ( \alpha , i ) = \sum_j \mM ( \alpha ; i , j ) .
\end{equation}
\end{definition}

The range $0\leq\alpha<1/\rho(\mathcal A)$ is sufficient for convergence of this path sum, but is not its convergence boundary. Beyond the first-hit sums' own convergence domain, rational continuation does not represent path weights. First-hit Katz stops at first arrival (first return for $i=j$). Both scores use outgoing row sums. Section~\ref{sec:threshold katz} continues the discussion of the all-walk resolvent.

\subsection{Extension Beyond the Spectral Threshold}\label{sec:threshold katz}
Let $\mQ=\mathcal A$ be a primitive adjacency matrix, so ${\lambda_1>|\lambda_2|}$.
The classical all-walk Katz connectivity
counts $\alpha$-discounted walks of all lengths and is well-defined for $0\leq\alpha < 1/\lambda_1$, with $ \lambda_1 $ the Perron eigenvalue of $ \mathcal{A}$, where the Neumann series $\sum_k \alpha^k \mathcal{A}^k$ converges absolutely. 
Past this threshold, that series diverges. We study its algebraic continuation and distinguish its signed scores from nonnegative walk sums.

We extend classical Katz connectivity to $1/\lambda_1<\alpha<1/|\lambda_2|$ using the deviation series together with the Perron term $\mathbf{H}/(1-\alpha\lambda_1)$. By \eqref{eq:transform}, their sum equals the all-walk Neumann series for $0\leq\alpha<1/\lambda_1$ and provides its algebraic continuation into this window. The deviation series alone is not equal to the all-walk resolvent.
For $ {1 / \lambda_1 < \alpha < 1 / | \lambda_2|}$, the $(i,j)$ entry of $\sum_{n=0}^\infty (\mathcal{A}^n - \lambda_1^{\,n}\mathbf{H}) \alpha^n$ is the discounted cumulative deviation between the length-$n$ walk count from $i$ to $j$ and its Perron-asymptotic value $\lambda_1^{\,n}\mathbf{H}(i,j) = \lambda_1^{\,n}\vx(i) \vy(j)$.
For this, we recall that 
$ \mathcal{A}^n - \lambda_1^{\,n}\mathbf{H} = ( \mathcal{A} - \lambda_1 \mathbf{H} )^n $ for $n\geq1$, and $ \mathcal{A} - \lambda_1 \mathbf{H} $ has spectral radius $ |\lambda_2| $. The $n=0$ term of the deviation series is $\mI-\mathbf{H}$.

The change of interpretation across the threshold $\alpha = 1/\lambda_1$ is a central point. Below the threshold, the resolvent is a nonnegative all-walk sum. Within the continuation window, its deviation series still converges because $\alpha|\lambda_2|<1$, while the Perron term $\mathbf{H}/(1-\alpha\lambda_1)$ changes sign. The continued score consists of both the Perron term and the deviation term, and therefore retains the Perron pole. In particular, close to the threshold from above, the negative Perron contribution can dominate the ranking.
This motivates the following definition.

\smallskip

\begin{definition}\label{def:generalized katz}
For $\alpha\in[0,1/|\lambda_2|)\setminus\{1/\lambda_1\}$, the 
\textit{extended Katz connectivity} of $ {\cal A} $ is 
$ [(\mathbf{I} - \alpha \mathcal{A})^{-1}-\mI]$,
and the \textit{extended Katz connectivity} of node $i\in\mathcal V$ is
    \[
        {\bf c}_{\rm ExtKatz} ( \alpha , i) := \sum_{j \in \mathcal{V}} [(\mathbf{I} 
        - \alpha \mathcal{A})^{-1}-\mI](i, j) .
    \]
\end{definition}
The extended Katz connectivity $ {\bf c}_{\rm ExtKatz}$ admits the following interpretations across the spectral threshold.
For $\alpha \in (0, 1/ \lambda_1)$, ${\bf c}_{\rm ExtKatz} ( \alpha)$ coincides with the classical Katz connectivity, with the absolute walk-counting interpretation. It does not coincide with the first-hit Katz connectivity in \eqref{eq:FHKatz}, which is a first-hit object; the two already differ on the directed $2$-cycle.
As $\alpha\to1/\lambda_1$ from either side, $(1-\lambda_1\alpha)\vc_{\rm ExtKatz}(\alpha)\to\vx(\vy^\top\one)$ by \eqref{eq:transform}.
And, for $\alpha \in (1/ \lambda_1,  1 / |\lambda_2|)$, 
${\bf c}_{\rm ExtKatz} ( \alpha ) $
ranks nodes by the sum of the signed Perron contribution and the discounted deviations from the Perron asymptotic. Section~\ref{sec:katz} illustrates the resulting role rankings on barbell graphs.

\section{Applications}\label{sec:Application}
We now illustrate the framework in three distinct settings. Subsection~\ref{sec:risk} keeps $\mQ$ deterministic and uses $\mM$ to model credit loss in an inter-bank system. In Subsection~\ref{sec:katz}, we apply the extended Katz connectivity directly and explore roles as $\alpha$ crosses $1/\lambda_1$. Subsection~\ref{sec:stochastic-app} introduces stochastic weights and questions the reliability of centrality rankings. 

\subsection{First-Hit Katz for Risk Dispersion in a Financial Network}\label{sec:risk}
Motivated by interbank networks such as \cite{eisenberg2001systemic}, we use a stylized linear first-arrival propagation model. An edge $(i,j)$ means that bank~$j$ is a creditor of bank~$i$, and $\mQ(i,j)\in(0,1)$ is its unscaled stress-transmission factor. We set $\mQ(i,j)=0$ off the edges and assume $\rho(\mQ)<1$. Then $\mM(i,j)$ sums the gains of paths stopped at their first arrival at $j$, and $(\vmu^\top\mM)(j)$ is the resulting stress at $j$ from initial shock $\vmu$. Row and column sums measure systemic footprint and vulnerability, respectively. Both include positive-time self-returns.

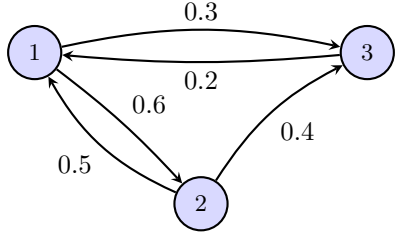
\begin{figure}[!t]
    \centering
    \begin{tikzpicture}[>=stealth, thick,
        vertex/.style={circle, draw=black, fill=blue!15, thick,
                       minimum size=20pt, inner sep=0pt, font=\small\sffamily}]
        \node[vertex] (1) at (0,0)      {$1$};
        \node[vertex] (2) at (2.2,-2) {$2$};
        \node[vertex] (3) at (4.4,0)    {$3$};
        \draw[->] (1) to[bend left=5]  node[above right] {$0.6$} (2);
        \draw[->] (2) to[bend left=18]  node[below left]{$0.5$} (1);
        \draw[->] (1) to[bend left=12]  node[above]      {$0.3$} (3);
        \draw[->] (3) to[bend left=5]  node[below]      {$0.2$} (1);
        \draw[->] (2) to[bend left=15]  node[below right]{$0.4$} (3);
    \end{tikzpicture}
    \caption{A three-bank network with loss pass-through fractions $\mQ(i,j)$ on the directed edges.}
    \label{fig:finnet}
\end{figure}

In Figure~\ref{fig:finnet}, direct pass-through from bank~$1$ to bank~$3$ is $\mQ(1,3)=0.30$. First-hit exposure also includes $1\to2\to3$ and repeated $1\leftrightarrow2$ feedback with round-trip gain $0.30$:
\begin{align*}
 \mM(1,3) &= (0.30+0.60\cdot0.40)\sum_{k=0}^\infty 0.30^k\\
         &= 0.54/0.70 \approx 0.771.
\end{align*}
This matches $\mR(1,3)/\mR(3,3)$ in Theorem~\ref{thm: sol invertible}. Indirect routing and feedback raise the direct pass-through by a factor of $2.57$. The full matrix in Table~\ref{tab:finM} identifies banks~$1$ and~$2$ as the largest sources of stress and bank~$3$ as the most vulnerable.

Under uniform scaling $\mQ\mapsto\gamma\mQ$, the all-paths Neumann series diverges at $\gamma^*=1/\rho(\mQ)\approx1.52$. However, the first-hit exposure, $
\mM(1,3;\gamma) =\frac{0.30\gamma+0.24\gamma^2}{1-0.30\gamma^2}, $
passes smoothly through $\gamma^*$, where the common pole in $\mR(1,3;\gamma)$ and $\mR(3,3;\gamma)$ cancels and the ratio tends to $x_1/x_3\approx3.29$. The first-hit sum remains convergent until $1/\sqrt{0.30}\approx1.83$, so it retains its path-based interpretation on an interval beyond the all-paths threshold, see Figure~\ref{fig:finnet_threshold}. The first-hit exposure remains finite beyond the threshold where the all-walk sum diverges. At sufficiently large scaling, some links amplify stress rather than pass on only a fraction of it.

\begin{figure}[t]
    \centering
    \includegraphics[width=0.5\columnwidth]{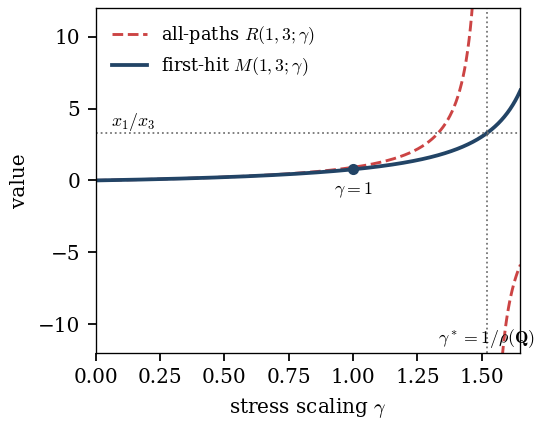}
    \caption{Under scaling $\gamma$, the resolvent entry $\mR(1,3;\gamma)$ has a pole at $\gamma^*\approx1.52$ and is only an algebraic continuation beyond it. The first-hit gain has finite limit $x_1/x_3\approx3.29$ there and its path sum converges until $\gamma\approx1.83$.}
    \label{fig:finnet_threshold}
\end{figure}

\subsection{Extended Katz on Barbell-Type Graphs}\label{sec:katz}
Barbell graphs reveal a tension between local density and global brokerage. Clique nodes are embedded in dense cores, while bridge nodes connect separate clusters. This makes them a natural test case for the extended Katz connectivity (and centrality) of Definition~\ref{def:generalized katz}. We use the parametrized family of \cite{breen2019computing}, which varies clique sizes, chain length, and number of chains independently.

\smallskip

\begin{definition}\label{def:barbell} For integers $k \geq 2$, $a \geq 3$ and $b,c \geq 1$, the \emph{barbell graph} $B(k,a,b,c)$ is the undirected graph on $ka + b + c$ vertices built as follows. Take $k$ disjoint copies of the path $P_a$, the \emph{chains}, and denote the two endpoints of chain $p$ by $\ell_p$ and $r_p$. Glue the chains at both ends by turning the left endpoints $\{\ell_1,\dots,\ell_k\}$ into a clique and the right endpoints $\{r_1,\dots,r_k\}$ into a clique. Finally, attach a complete graph $K_b$ by joining each of its vertices to every left endpoint $\ell_p$, and a complete graph $K_c$ by joining each of its vertices to every right endpoint $r_p$.
\end{definition}

The left and right ends of $B(k,a,b,c)$ are cliques $K_{k+b}$ and $K_{k+c}$, connected by $k$ internally disjoint paths. We distinguish three roles, shown in Figure~\ref{fig:barbell_example}: \emph{clique} nodes, \emph{attach} nodes, and \emph{bridge} nodes. Since the graph is connected and non-bipartite, ${\cal A}$ is primitive and $\lambda_1>|\lambda_2|$, so the continuation window $(1/\lambda_1,1/|\lambda_2|)$ is non-empty.

\begin{figure}[t!]
    \centering
    \begin{tikzpicture}[
        cnode/.style={circle, draw=black, fill=blue!18,   thick, minimum size=13pt, inner sep=0pt},
        anode/.style={circle, draw=black, fill=green!30,  thick, minimum size=13pt, inner sep=0pt},
        bnode/.style={circle, draw=black, fill=orange!35, thick, minimum size=13pt, inner sep=0pt},
        cedge/.style={draw=black!35, thin},
        pedge/.style={draw=black, thick}
    ]
    \coordinate (Lg1) at (-0.65, 0.80); \coordinate (Lg2) at (-0.65,-0.80);
    \coordinate (Lb1) at (-1.56, 1.38); \coordinate (Lb2) at (-2.70, 1.07);
    \coordinate (Lb3) at (-3.20, 0.00); \coordinate (Lb4) at (-2.70,-1.07);
    \coordinate (Lb5) at (-1.56,-1.38);
    \coordinate (b11) at ( 1.10, 0.80); \coordinate (b12) at ( 2.70, 0.80);
    \coordinate (b21) at ( 1.10,-0.80); \coordinate (b22) at ( 2.70,-0.80);
    \coordinate (Rg1) at ( 4.43, 0.71); \coordinate (Rg2) at ( 4.43,-0.71);
    \coordinate (Rc1) at ( 5.77, 1.14); \coordinate (Rc2) at ( 6.60, 0.00);
    \coordinate (Rc3) at ( 5.77,-1.14);
    \fill[blue!5] (-1.85,0) ellipse (1.75 and 1.95);
    \fill[blue!5] ( 5.50,0) ellipse (1.55 and 1.70);
    \foreach \u/\v in {Lg1/Lg2,Lg1/Lb1,Lg1/Lb2,Lg1/Lb3,Lg1/Lb4,Lg1/Lb5,%
        Lg2/Lb1,Lg2/Lb2,Lg2/Lb3,Lg2/Lb4,Lg2/Lb5,%
        Lb1/Lb2,Lb1/Lb3,Lb1/Lb4,Lb1/Lb5,%
        Lb2/Lb3,Lb2/Lb4,Lb2/Lb5,Lb3/Lb4,Lb3/Lb5,Lb4/Lb5}
        \draw[cedge] (\u)--(\v);
    \foreach \u/\v in {Rg1/Rg2,Rg1/Rc1,Rg1/Rc2,Rg1/Rc3,%
        Rg2/Rc1,Rg2/Rc2,Rg2/Rc3,Rc1/Rc2,Rc1/Rc3,Rc2/Rc3}
        \draw[cedge] (\u)--(\v);
    \draw[pedge] (Lg1)--(b11)--(b12)--(Rg1);
    \draw[pedge] (Lg2)--(b21)--(b22)--(Rg2);
    \node[anode] at (Lg1){}; \node[anode] at (Lg2){};
    \node[cnode] at (Lb1){}; \node[cnode] at (Lb2){}; \node[cnode] at (Lb3){};
    \node[cnode] at (Lb4){}; \node[cnode] at (Lb5){};
    \node[bnode] at (b11){}; \node[bnode] at (b12){};
    \node[bnode] at (b21){}; \node[bnode] at (b22){};
    \node[anode] at (Rg1){}; \node[anode] at (Rg2){};
    \node[cnode] at (Rc1){}; \node[cnode] at (Rc2){}; \node[cnode] at (Rc3){};
    \node at (-1.85,-2.5) {$K_{k+b}=K_7$};
    \node at ( 5.50,-2.5) {$K_{k+c}=K_5$};
    \node[cnode, minimum size=11pt] at (0.0,-3.2){};
        \node[anchor=west] at (0.25,-3.2){ clique};
    \node[anode, minimum size=11pt] at (2.0,-3.2){};
        \node[anchor=west] at (2.25,-3.2){ attach};
    \node[bnode, minimum size=11pt] at (4.0,-3.2){};
        \node[anchor=west] at (4.25,-3.2){ bridge};
    \end{tikzpicture}
    \caption{Barbell graph $B(2,4,5,3)$ with end cliques $K_7$ and $K_5$. Colors distinguish clique, attach, and bridge nodes.}
    \label{fig:barbell_example}
\end{figure}
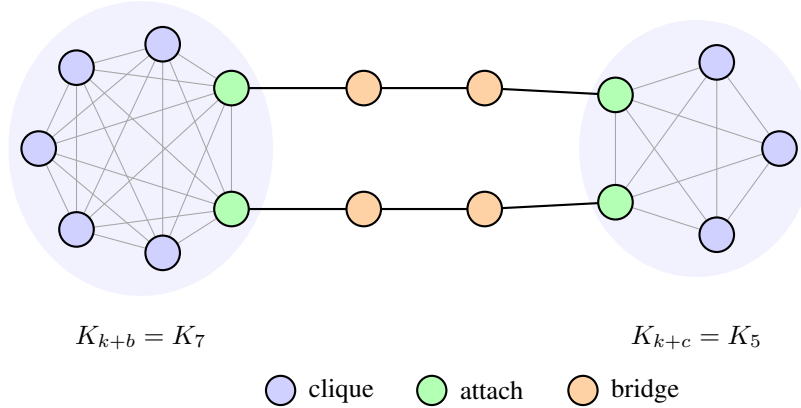

We evaluate extended Katz connectivity in two regimes: below the threshold, $\alpha=0.9/\lambda_1$, where it coincides with classical Katz centrality, and inside the continuation window, at the midpoint of $(1/\lambda_1,1/|\lambda_2|)$. We enumerate all $300$ barbell graphs with $k\in\{2,\dots,5\}$, $a\in\{4,\dots,8\}$, and $3\leq b<c\leq8$. The restriction $b<c$ excludes mirror duplicates, since $B(k,a,b,c)$ and $B(k,a,c,b)$ are isomorphic. All retained graphs have distinct degree sequences, confirming that no other isomorphic pairs occur. For each graph, we rank the three structural roles by the arithmetic mean of their node centralities. Figure~\ref{fig:katz_rank_freq} reports the resulting rank frequencies.

\begin{figure}[t]
\centering
\includegraphics[width=0.6\columnwidth]{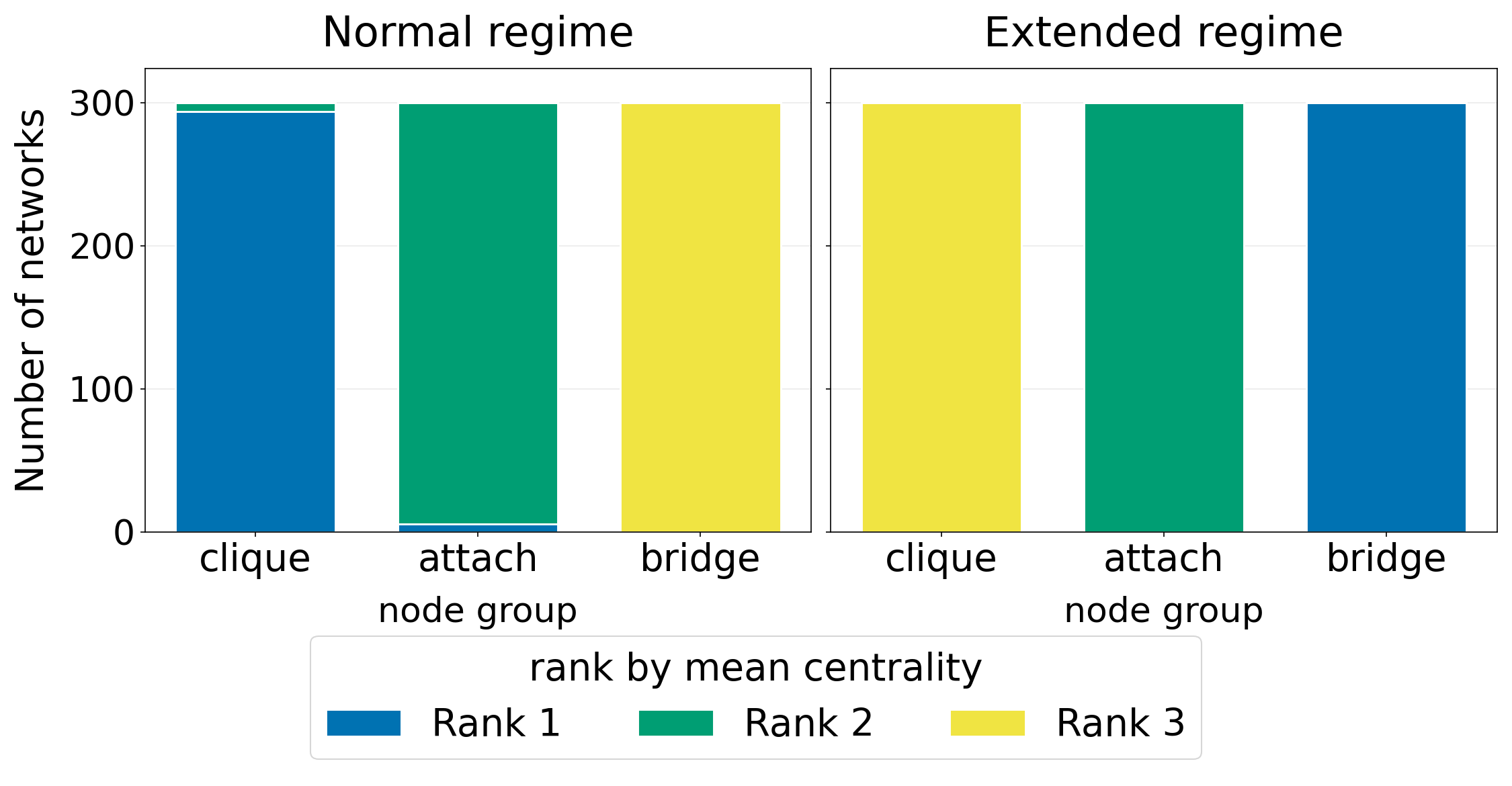}
\caption{Ranks by mean centrality within each role across $300$ non-isomorphic barbell graphs. Left: classical Katz centrality below the threshold, at $\alpha=0.9/\lambda_1$.
Right: extended Katz centrality above the threshold, at the midpoint of $(1/\lambda_1,1/|\lambda_2|)$.}
\label{fig:katz_rank_freq}
\end{figure}

Below the threshold, the ranking is core-dominated: clique nodes have the highest mean centrality in $294$ of the $300$ graphs, attach nodes lead in the remaining $6$, and bridge nodes are always last. This matches classical Katz centrality, which rewards many short walks in dense regions. At the midpoint of the continuation window, bridge nodes have the highest role mean in all $300$ graphs, attach nodes rank second, and clique nodes rank third. This completely reverses the below-threshold ordering in $294$ graphs. This inversion illustrates the analytic continuation in Section~\ref{sec:threshold katz}. The continued score includes both a negative Perron contribution and discounted deviations; the reversal cannot be attributed to bridges having the largest positive deviations. It concerns role averages at the evaluated midpoint, not individual-node rankings or the entire window.

In a financial-network interpretation, the cliques represent strongly connected lending markets, while bridge nodes represent intermediaries between markets. Below the threshold, Katz centrality identifies banks embedded in dense mutual-lending clusters~\citep{acemoglu2015systemic, eisenberg2001systemic}. Above the threshold, the scores can be negative, so the observed reversal of role means alone does not establish greater systemic importance of bridge banks. For related role-based approaches in financial networks, see \cite{franssen2026practicalguideinterpretablerolebased}.

\subsection{Reliability of Centrality Ranking under Stochastic Weights}\label{sec:stochastic-app}
We test whether the mean centrality ranking on $B(2,4,5,3)$ remains reliable when edge weights are random. A target mean $\mu(i,j)\sim{\rm Unif}(0,1)$ is drawn once for each edge and held fixed. Every edge traversal then draws ${\cal W}(i,j) \sim {\rm Gamma}(\kappa, \kappa/ \mu(i,j))$ (shape, rate), so that
\[
    \ev[{\cal W}(i,j)]=\mu(i,j), \quad \text{ and } \quad
    \var[\mathcal W(i,j)]=\frac{\mu(i,j)^2}{\kappa}.
\]
We compare a low variance regime, $\kappa=25$ (edge CV $0.20$), to a high variance one, $\kappa=0.5$ (edge CV $1.41$). Throughout, we assume each path carries its own independent draws, so distinct path weights are independent. Consequently, the variance can be calculated exactly (Theorem~\ref{thm:variance}).

\begin{figure}[t]
    \centering
        \includegraphics[width=0.6\linewidth]{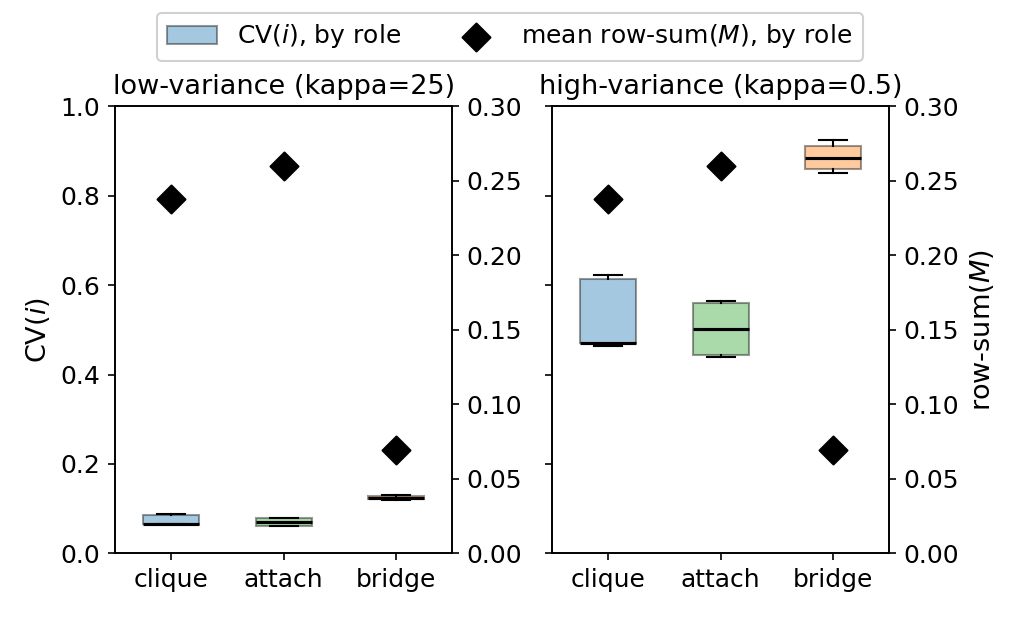}
    \caption{Coefficient of variation shown by role as boxplots (i.e., spread across nodes of the same role) on the left axis, with the node's mean centrality, $\sum_j \mM(i,j)$, shown as fixed diamonds on the right axis.}
    \label{fig:CV_by_role}
\end{figure}

A node-level measure of uncertainty is the coefficient of variation 
\[ {\rm CV}(i) := \sqrt{\sum_j \mV(i,j)} / \sum_j \mM(i,j),\]
where we compare roles (bridge, attach, clique) through their role-means.
Figure~\ref{fig:CV_by_role} shows that bridge nodes have the lowest mean centrality but the largest CV, especially under high variance. Considered alone, this \textit{suggests} that their position in the ranking is the least reliable. Ranking reliability, however, is a property of a \textit{pair}, i.e., it depends not only on a node's own variance but also on how that variance compares to that of other nodes.

Since roles are compared through their means, i.e., $\overline{c_G} := |G|^{-1} \sum_{i \in G} c_i$ for a role $G$ and a per-node centrality score of $c_i$, for disjoint roles $G, H$ the paired difference $d_{GH}:= \overline{c_G} - \overline{c_H}$ satisfies,
\[
    \overline{d_{GH}} = |G|^{-1} \sum_{i \in G} \sum_j \mM(i,j) - |H|^{-1} \sum_{k \in H} \sum_j \mM(k,j)
\]
for its mean, and
\[
    \var(d_{GH}) = |G|^{-2} \sum_{i\in G}\sum_j \mV(i,j) + |H|^{-2} \sum_{k\in H}\sum_j \mV(k,j)
\]
for its variance. We summarize a pair by the gap-to-noise ratio
\[GTN_{GH} := |\overline{d_{GH}}|/ \sqrt{\var(d_{GH})}\]
across 5,000 Monte Carlo replications, truncated at path length 5 for numerical purposes. Compared with the exact $\mM$ and $\mV$, the simulated means deviate by at most 0.3\% ($\kappa=25$) and 2.6\% ($\kappa=0.5$), and the variances by at most 4.6\% and 7.0\%.

\begin{figure}[t]
    \centering
    \includegraphics[width=0.6\linewidth]{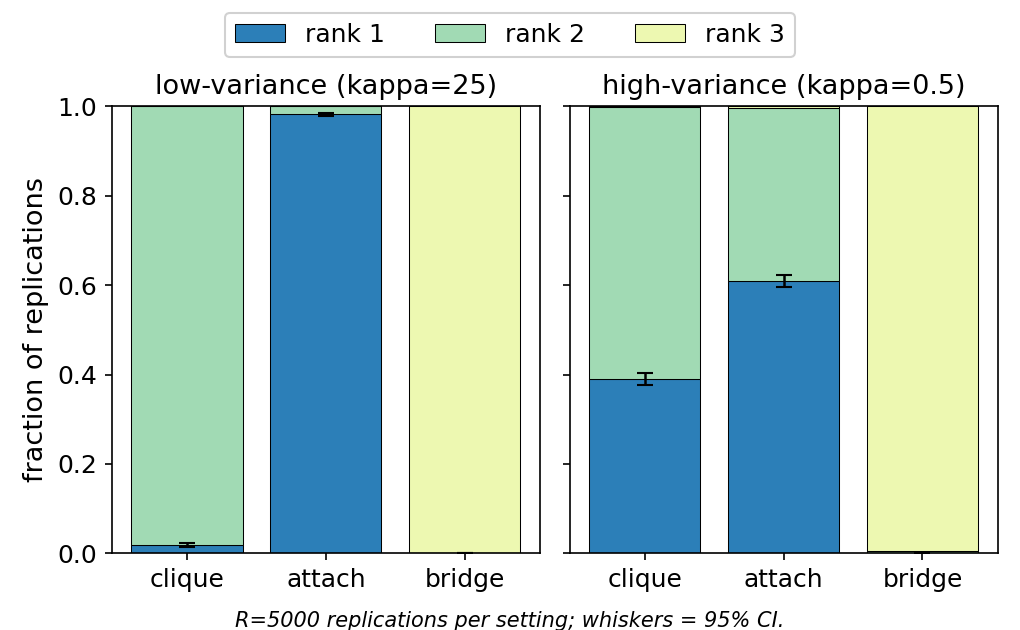}
    \caption{Empirical rank-frequency distribution of first-hit
    exposure by structural role, from $R=5000$ realized weight-matrix
    draws per setting (whiskers: 95\% CI on the rank-1 fraction). Under
    low variance, the ranking is essentially fixed.
    Under high variance, bridge nodes remain
    reliably last, while clique and attach swap rank 1/2 in close to half of all replications.}
    \label{fig:rank_reliability}
\end{figure}

Figure~\ref{fig:rank_reliability} shows what ${\rm CV}(i)$, read node-by-node, does not reveal. In the low-variance regime, attach outranks clique in 98.2\% $\pm$ 0.2\% of replications and bridges are ranked last in all of them. The GTN equals 2.07 for clique vs. attach, 22.6 for clique vs. bridge, and 19.1 for attach vs. bridge. In the high-variance regime, the clique-attach GTN drops to 0.29, and clique ranks above attach in 39.1\% $\pm$ 0.7\% of replications. Against the bridges, the GTN remains well above one (3.17 and 2.69), and bridges are ranked last in 99.5\% of replications. The nodes with the largest individual CV thus hold the most \textit{reliable} rank, because their gap to the other roles dominates the combined noise. Increased variability does not overturn the mean ranking: attach is ranked first most often in both settings, while variability mainly blurs the clique-attach order.

In summary, $\mM$ yields the mean ranking, and $\mV$ provides each node's exposure uncertainty in closed form. Ranking reliability is a pairwise quantity that a node's own variance cannot determine. Under path-indexed coupling, it follows from $\mM$ and $\mV$ of both nodes, and under couplings with shared weight realizations, covariances between first-hit weights could enter as well.

\section{Further Research}\label{sec:Conclusion}
This paper presents a unified theoretical framework for analyzing network connectivity based on multiplicative path weights, in which the weight of a directed path is defined as the product of its edge weights. 
Further research will be on resilient network clustering, where we will elaborate on the moment expressions for network connectivity in the case of stochastic edge weights.

\section*{Acknowledgments}
A. Zocca and R. van der Burg are supported by the NWO VI.Vidi.233.247 grant \textit{Power Network Optimization in the Age of Climate Extremes}~\url{https://doi.org/10.61686/GOOEL09973}.

\bibliographystyle{unsrtnat}        
\bibliography{references}           




\newpage
\begin{appendices}
\section{Proof of Theorem~\ref{thm: sol invertible}}\label{app:proof_thm_1}
Multiplying \eqref{eq:basic recursion} by $ ( \mI - \mQ  )^{-1}$ yields
\[
    \mM = ( \mI - \mQ  )^{-1} 
    \mQ ( \mI - \matdg{\mM}) .
\]
Observing that
\[
        (\mI-\mQ)^{-1}\mQ = (\mI-\mQ)^{-1} - \mI
\]
follows from the fact that $(\mI-\mQ)^{-1}(\mI-\mQ)=\mI$, we  now arrive at 
\begin{align}\label{eq:op2}
    \mM & = (( \mI - \mQ  )^{-1} - I )
     ( \mI - \matdg{\mM})  \nonumber \\
     & = ( \mR - \mI )
     ( \mI - \matdg{\mM}).
\end{align}
The diagonals thus satisfy
\begin{equation}\label{eq:diagonals}
\mM(i,i) = (\mR(i,i)-1)(1-\mM(i,i)),
\end{equation}
which rewrites as
\[
\mR(i,i)\mM(i,i) = \mR(i,i)-1,
\]
from which it follows that $\mR(i,i) \neq 0$ for all $i \in {\cal V}$, otherwise $0=-1$. From \eqref{eq:diagonals}, the diagonal elements of $\mM$ can be obtained explicitly as
\begin{equation}\label{diagonals invertible regime}
    \mM ( i, i ) =  1 - \frac{1}{\mR(i,i)} .
\end{equation}
Therefore,
\[
    ( \mI - \matdg{\mM})( i , i ) = \frac{1}{\mR(i,i)}.
\]
Combining this with \eqref{eq:op2} gives the following expression for the off-diagonal terms:
\[
    \mM (i,j) = \frac{\mR(i,j)}{\mR(j,j)},
\]
which, combined with~\eqref{diagonals invertible regime}, concludes the proof.

\section{Generalized Inverse Approach for Perron-Critical Case}\label{sec:gen_inverse_approach_appendix}
    Since standard matrix inversion fails for $\mI - \mQ$ in this regime, we rely on a \textit{generalized inverse} of $\mI-\mQ$ to solve the system of equations. A generalized inverse of a matrix $\mA$ is a matrix $\mA^\dagger$ such that $\mA\mA^\dagger\mA = \mA$. Building on the works of Hunter \citep{hunter1, hunter2, markovrenewal1969hunter}, we can guarantee the existence of a generalized inverse of $\mI - \mQ$ using the following lemma.

\smallskip
\begin{lemma}\label{lem: Generalized Inverse}
Let $\mA\in\R^{n\times n}$ with $\operatorname{rank}(\mA)=n-1$, and let $\vx,\vy\neq{\bf0}$ satisfy $\mA\vx={\bf0}$ and $\vy^\top\mA={\bf0}^\top$. For $\vu,\vt\in\R^n$, the matrix $\mA+\vu\vt^\top$ is invertible if and only if
\begin{equation}
\vy^\top\vu\neq0\qquad\text{and}\qquad\vt^\top\vx\neq0 .\nonumber
\end{equation}
In that case $\mG:=(\mA+\vu\vt^\top)^{-1}$ is a generalized inverse of $\mA$, that is, $\mA\mG\mA=\mA$.
\end{lemma}
\begin{proof}
    If $\vy^\top\vu=0$, then $\vy^\top(\mA+\vu\vt^\top)={\bf0}^\top$; if $\vt^\top\vx=0$, then $(\mA+\vu\vt^\top)\vx={\bf0}$. In both cases $\mA+\vu\vt^\top$ is singular, which proves necessity. For sufficiency, let $(\mA+\vu\vt^\top)\vz={\bf0}$. Premultiplying by $\vy^\top$ gives $(\vy^\top\vu)(\vt^\top\vz)=0$, hence $\vt^\top\vz=0$ and $\mA\vz={\bf0}$. Since $\ker(\mA)=\operatorname{span}\{\vx\}$ by the rank condition, $\vz=c\vx$, and $0=\vt^\top\vz=c\,\vt^\top\vx$ gives $c=0$. Thus $\mA+\vu\vt^\top$ is invertible.

Premultiplying $(\mA+\vu\vt^\top)\mG=\mI$ by $\vy^\top$ and using $\vy^\top\mA={\bf0}^\top$ gives $\vt^\top\mG=\vy^\top/(\vy^\top\vu)$. Substituting this into $(\mA+\vu\vt^\top)\mG=\mI$ yields
\[
\mA\mG=\mI-\frac{\vu\vy^\top}{\vy^\top\vu},
\]
and multiplying on the right by $\mA$ gives $\mA\mG\mA=\mA$.
\end{proof}
The rank condition and both pairing conditions are needed. For $\mA={\bf0}\in\R^{2\times2}$ and $\vx=\vy=\vu=\vt=\ve_1$, both pairings equal one, yet $\mA+\vu\vt^\top=\ve_1\ve_1^\top$ is singular; and for $\mA=\mI-{\cal A}$ with ${\cal A}$ the two-cycle, $\vx=\vy=\one$, the choices $\vu=(1,-1)^\top,\vt=\one$ or $\vu=\one,\vt=(1,-1)^\top$ violate one pairing each and give singular perturbations.

\begin{corollary}\label{cor:canonical}
Let $\rho(\mQ)=1$, and let $\vx,\vy>{\bf0}$ satisfy $\mQ\vx=\vx$, $\vy^\top\mQ=\vy^\top$ and $\vy^\top\vx=1$. Then
\begin{equation}\label{eq:funda}
\fundmat:=(\mI-\mQ+\vx\vy^\top)^{-1}
\end{equation}
exists, is a generalized inverse of $\mI-\mQ$, and satisfies $\fundmat\vx=\vx$ and $\vy^\top\fundmat=\vy^\top$. Moreover, the group inverse of $\mI-\mQ$ is
\[
(\mI-\mQ)^{\#}=\fundmat-\vx\vy^\top .
\]
\end{corollary}
\begin{proof}
By the Perron--Frobenius theorem, eigenvalue one of $\mQ$ is algebraically simple, so $\operatorname{rank}(\mI-\mQ)=n-1$. Lemma~\ref{lem: Generalized Inverse} with $\mA=\mI-\mQ$, $\vu=\vx$, $\vt=\vy$ gives existence and the generalized-inverse property, since $\vy^\top\vx=1\neq0$. From $(\mI-\mQ+\vx\vy^\top)\vx=\vx$ and $\vy^\top(\mI-\mQ+\vx\vy^\top)=\vy^\top$ we obtain $\fundmat\vx=\vx$ and $\vy^\top\fundmat=\vy^\top$. With $\mH=\vx\vy^\top$ these identities give $(\mI-\mQ)\fundmat=\fundmat(\mI-\mQ)=\mI-\mH$. Since $(\mI-\mQ)\mH=\mH(\mI-\mQ)={\bf0}$ and $\mH^2=\mH$, the matrix $\mG=\fundmat-\mH$ satisfies $(\mI-\mQ)\mG=\mG(\mI-\mQ)=\mI-\mH$, $(\mI-\mQ)\mG(\mI-\mQ)=\mI-\mQ$ and $\mG(\mI-\mQ)\mG=\mG$, which are the defining identities of the group inverse.
\end{proof}

We record three properties of $\fundmat$. (i) If $\mQ=\mP$ is an irreducible stochastic matrix, then $\vx=\one$, $\vy=\vpi$, and $\fundmat$ is the fundamental matrix $\mZ=(\mI-\mP+\mPi)^{-1}$ of \cite{kemeny1976finite}; aperiodicity is not required. (ii) If $\mQ$ is in addition primitive, then $\fundmat-\mH=\sum_{n\geq0}(\mQ^n-\mH)$, with geometric convergence; in the periodic case the series must be replaced by its Abel limit $\lim_{\gamma\uparrow1}\sum_{n\geq0}\gamma^n(\mQ^n-\mH)$, which exists and equals $\fundmat-\mH$ by the argument leading to \eqref{eq:laurent}. (iii) The normalization $\vy^\top\vx=1$ only shortens the formulas; for arbitrarily scaled Perron vectors, the projector is $\vx\vy^\top/(\vy^\top\vx)$. For $\mQ=\mP\circ\mW$, the Perron vectors of $\mQ$ cannot in general be obtained from those of $\mP$ and $\mW$.

\subsection{Alternative derivation of \eqref{sol_perron_critical}}\label{proofmsingular}
Let $\mQ$ satisfy the assumptions of Corollary~\ref{cor:canonical}. By Theorem~\ref{thm:sol_singular}, $\mM$ is finite, and by Lemma~\ref{lem:diagonals non invertible}, $\matdg{\mM}=\mI$. Equation \eqref{eq:basic recursion} therefore reduces to
\begin{equation}\label{I-Q non-invertible not solved}
(\mI-\mQ)\mM={\bf0}.
\end{equation}
The derivation below solves \eqref{I-Q non-invertible not solved} with the generalized inverse $\fundmat$ of Corollary~\ref{cor:canonical}. It yields the same formula, but it presupposes finiteness, which the Perron-transform proof establishes, and it uses the left Perron vector, which that proof does not need.

For this proof, we use the following result from linear algebra (see, for instance, \cite{hunter1, hunter3}). Using the generalized inverse $\mA^\dagger$ of a matrix $\mA$, it is easily seen that $\mX = \mA^\dagger \mB + (\mI - \mA^\dagger \mA)\mU$, where $\mU$ is an arbitrary matrix, solves $\mA \mX = \mB$, provided that $\mA^\dagger$ is \textit{consistent} with $\mA$, which means $(\mI - \mA \mA^\dagger)\mB = 0$.

In our case \eqref{I-Q non-invertible not solved}, $\fundmat$ is trivially consistent, since $\mB = 0$. Hence, the solution to \eqref{I-Q non-invertible not solved} is
\begin{equation}\label{eq:m1}
    \mM = (\mI - \fundmat (\mI -\mQ)) \mU,
\end{equation}
where $\mU$ is an arbitrary matrix. In fact, $\mU$ is not an arbitrary matrix at all and is bound to the diagonals $\matdg{\mM}$ we derived earlier. Seeing these diagonals as ``constraints'', we can derive what matrix $\mU$ recovers the correct diagonals.

Now, from
\[
    \fundmat (\mI - \mQ + \vx\vy^\top) = \mI,
\]
we get
\[
    \fundmat (\mI - \mQ) = \mI - \fundmat \vx\vy^\top,
\]
and hence,
\[
    \mI - \fundmat(\mI - \mQ) = \fundmat \vx\vy^\top.
\]
Therefore, we obtain
\begin{align}\label{eq U not solved}
    \mM = \fundmat \vx\vy^\top \mU.
\end{align}
Since $\fundmat \vx$ is a column vector and $\vy^\top \mU$ is a row vector,
\[
    \mM(i,i) = (\fundmat \vx)_i \cdot (\vy^\top \mU)_i,
\]
and thus,
\[
    (\vy^\top \mU)_i = \frac{1}{(\fundmat \vx)_i}.
\]
In matrix notation, this means that
\begin{align}\label{U matrix}
    \vy^\top \mU = \one^\top \diag(\fundmat \vx)^{-1}.
\end{align}

Substituting \eqref{U matrix} back into \eqref{eq U not solved}, we obtain the final result,
\begin{align*}
    \mM = \fundmat \vx \cdot \one^\top \diag(\fundmat \vx)^{-1}.
\end{align*}
Entry-wise, this reads as
\[
    \mM(i,j) = \frac{(\fundmat \vx)_i}{(\fundmat \vx)_j},
\]
similar to the expressions found in the invertible case. Note that $(\mI - \mQ + \vx \vy^\top)\vx = \vx$ and thus, since $\vy^\top \vx = 1$, we find $\fundmat \vx = \vx$. Thus, the expression reduces to
\[
    \mM(i,j) = \frac{\vx_i}{\vx_j},
\]

and, indeed, for the diagonals, we find
\[
    \mM(i,i) = \frac{\vx_i}{\vx_i} = 1
\]
for all nodes $i \in {\cal V}$. This concludes the proof.

\section{Proof of Theorem~\ref{thm:sdq}}\label{app: proof theorem 3}
Let $\mP$ be the transition matrix of a finite irreducible Markov chain, with stationary law $\vpi>{\bf0}$, $\mPi=\one\vpi^\top$, fundamental matrix $\mZ=(\mI-\mP+\mPi)^{-1}$ and deviation matrix $\mD=\mZ-\mPi$. No aperiodicity is assumed.

{\bf Algebraic proof of \eqref{eq:mfpt}.} On a punctured real neighborhood of $\gamma=1$ that contains no reciprocal eigenvalue of $\mP$ other than $1$, we have $\mM(i,j;\gamma)=\mR(i,j;\gamma)/\mR(j,j;\gamma)$ for $i\neq j$ and $\mM(j,j;\gamma)=1-1/\mR(j,j;\gamma)$. Write $t=1-\gamma$ and $\partial_\gamma\mM(i,j;1^-)$ as a shorthand notation for $\left.\frac{d}{d\gamma }\mM(i,j;\gamma )\right|_{\gamma=1^-}$. By \eqref{eq:laurent},
\[
\mR(i,j;\gamma)=\frac{\vpi(j)}{t}+\mD(i,j)+{\cal O}(t).
\]
Hence, for $i\neq j$,
\begin{align*}
\mM(i,j;\gamma)&=\frac{\vpi(j)+t\,\mD(i,j)+{\cal O}(t^2)}{\vpi(j)+t\,\mD(j,j)+{\cal O}(t^2)}\\
&=1+\frac{t}{\vpi(j)}\big(\mD(i,j)-\mD(j,j)\big)+{\cal O}(t^2),
\end{align*}
and, for the diagonal,
\begin{align*}
\mM(j,j;\gamma)&=1-\frac{t}{\vpi(j)+t\,\mD(j,j)+{\cal O}(t^2)}\\
&=1-\frac{t}{\vpi(j)}+{\cal O}(t^2).
\end{align*}
Since $\mM(i,j;\cdot)$ is analytic for $\gamma \uparrow 1$, these expansions identify its Taylor coefficients there. With $\partial_\gamma=-\partial_t$ and $\mD(j,j)-\mD(i,j)=\mZ(j,j)-\mZ(i,j)$ (the $\mPi$-terms cancel), we obtain $\partial_\gamma\mM(i,j;1^-)=(\mZ(j,j)-\mZ(i,j))/\vpi(j)$ for $i\neq j$ and $\partial_\gamma\mM(j,j;1^-)=1/\vpi(j)$. These are the classical expressions for $\mathbb E_i\tau_j^+$ \citep{kemeny1976finite} and, on the diagonal, Kac's formula for the mean return time. The derivative is two-sided.

{\bf [Probabilistic Proof:]}
The weight of a first-hit path $\sigma\in\mathcal{P}(i,j)$ of length $l(\sigma)$ under the scaled matrix $\gamma\bm{P}$ equals $\gamma^{l(\sigma)}\prod_k \bm{P}(i_k,j_k)$. Summing over $\mathcal{P}(i,j)$ therefore gives
\begin{equation}
\bm{M}(i,j; \gamma)
  =  \! \! \! \sum_{\sigma\,\in\,\mathcal{P}(i,j)}
\gamma^{l(\sigma)}\prod_{k=1}^{l(\sigma)}\bm{P}(i_k,j_k)
  \;=\; \mathbb{E}\!\left[\gamma^{\tau_{ij}}\right], \nonumber
\end{equation}
where $\tau_{ij}:=\inf\{n\geq 1: X_n=j \mid X_0 = i\}$ is the first-passage time of the chain $(X_n)$ governed by $\bm{P}$ and started at $i$.
In fact, $\gamma\mapsto\bm{M}(\gamma;i,j)$ is the probability generating function of $\tau_{ij}$,
and it holds
\begin{align}
  \bm{M}(i,j ; 1)
&=\mathbb{E}\!\left[1^{\tau_{ij}}\right]
     = \mathbb{P}(\tau_{ij}<\infty)
     = 1, \label{eq:pgf-value}\\[4pt]
  \left.\frac{{d}}{{d}\gamma}
    \bm{M}(i,j; \gamma)\right|_{\gamma=1^-} \! \!\ \! \! \! 
    & = \left. \mathbb{E} \!\left[\tau_{ij}\gamma^{\tau_{ij}-1}\right] \right |_{\gamma=1}
     \! \! \!
     = \mathbb{E}[\tau_{ij}]
     =  \! \bm{L}(i,j). \label{eq:pgf-deriv}
\end{align}
Equation~\eqref{eq:pgf-value} is the statement of Theorem~\ref{thm:sol_singular} restricted to $\bm{Q}=\bm{P}$: the multiplicative matrix $\bm{M}$ equals $\bm{1}\bm{1}^\top$ for any ergodic Markov chain, encoding the certainty of hitting every state. 
Equation~\eqref{eq:pgf-deriv} is precisely the result from the first part of the theorem:
the classical, additive mean first-passage time $\bm{L}(i,j)=\mathbb{E}[\tau_{ij}]$ emerges only upon differentiation. 

\section{Proof of Lemma~\ref{lem:pos_corr}}\label{app:pos_cor_lem}
    To distinguish multiple occurrences of the same edge along a path, we write $(e,k)\in\sigma$ if edge $e$ is the $k$-th edge traversed by $\sigma$. Hence, $\sigma\cap\sigma'$ denotes the set of edge-traversal pairs shared by $\sigma$ and $\sigma'$. For edge $ e $, we write $ w ( e , k ) $ for the weight of edge $ e$ at the $k$th traversal (so as not to confuse it with the notation $ {\cal W} ( i , j ) $ for the edge weight of edge $ e = ( i , j ) $).
    
Under Assumption~\ref{ass:edge_share}, a shared edge traversal may use the same realization of the stochastic edge weight on both paths, while all other realizations are independent. Therefore,
\begin{align*}
& \mathbb{E} [ Y_\sigma Y_{\sigma'} ]
\\
& = 
\mathbb{E} \left [ \Pi_{ (e, k )  \in \sigma \cap \sigma'} (w ( e , k ) )^{2}
\Pi_{(e , k ) \in (\sigma \cap \sigma' )^c } w ( e , k ) \right ]
\\
& = 
\Pi_{ ( e , k )  \in \sigma \cap \sigma'}  \mathbb{E} \left [  ( w ( e , k ) )^{2 }  \right ]
\Pi_{ ( e , k ) \in (\sigma \cap \sigma' )^c}
\mathbb{E}^2 \left [ w( e , k ) \right ]
\end{align*}
We can now use the fact that 
\[
\Pi_{ ( e , k ) \in \sigma \cap \sigma'}  
\mathbb{E} \left [  ( w ( e, k ) )^2  \right ] \geq 
\Pi_{ ( e , k ) \in \sigma \cap \sigma'}  
\mathbb{E}^2 \left [  ( w ( e, k  ) )  \right ],
\]
to show that
\[
\mathbb{E} [ Y_\sigma Y_{\sigma'} ]
\geq 
\mathbb{E} [ Y_\sigma ] 
\mathbb{E}[Y_{\sigma'} ] .
\]
Equivalently,
$
\operatorname{cov}\left(Y_\sigma,Y_{\sigma'}\right)\geq 0.
$

If the stochastic weight realizations used by $\sigma$ and $\sigma'$ are independent, no common realization occurs in the product, and the above inequality holds with equality. This concludes the proof.

{\section{Results for the Financial Network}\label{app:res_fin}
\begin{table}[h!]
    \centering
    \caption{Network-effective exposure $\mM$ at $\gamma = 1$ for the
    network of Figure~\ref{fig:finnet}. Entry $(i,j)$ is the loss landing
    on bank $j$ from a unit shock at bank $i$; the margins report each
    bank's systemic footprint (row sum) and vulnerability (column sum).}
    \begin{tabular}{c|ccc|c}
                     & $j{=}1$ & $j{=}2$ & $j{=}3$ & footprint \\ \hline\hline
        $i{=}1$      & 0.41 & 0.64 & 0.77 & 1.82 \\
        $i{=}2$      & 0.58 & 0.37 & 0.79 & 1.74 \\
        $i{=}3$      & 0.20 & 0.13 & 0.15 & 0.48 \\ \hline
        vulnerability & 1.19 & 1.14 & 1.71 &      \\
    \end{tabular}
    \label{tab:finM}
\end{table}}
\end{appendices}

\end{document}